\documentclass{article}
\usepackage[utf8]{inputenc}
\usepackage[english]{babel}

\usepackage{units}
\usepackage{tabularx}

\usepackage{framed}
\usepackage{graphics}
\usepackage{xcolor}
\usepackage{tikz}
\usepackage{subcaption}
\usepackage{multirow}
\usetikzlibrary{shapes.arrows}

\tikzstyle{normalNode}=[draw=black, fill=gray, circle, minimum size=1em]
\tikzstyle{labeledNode}=[normalNode, fill=white, inner sep=0.5mm]
\tikzstyle{facNode}=[normalNode, fill=white, rectangle]
\tikzstyle{normalEdge}=[black, very thick, >=stealth]
\tikzstyle{dashedEdge}=[dashed, black, very thick, >=stealth]

\usepackage{natbib}
\setcitestyle{authoryear}
\usepackage{hyperref}

\usepackage{amssymb,amsmath,amsthm}

\usepackage[algoruled,vlined]{algorithm2e}

\SetKwFor{ForAll}{for all}{do}{endfall}

\usepackage[capitalise]{cleveref}

\newtheorem{theorem}{Theorem}
\newtheorem{lemma}[]{Lemma}
\newtheorem{corollary}[]{Corollary}
\newtheorem{claim}[]{Claim}

\usepackage{stackrel}

\newcommand{\FLext}{Facility Location}

\newcommand{\UFLext}{Uncapacitated Facility Location}
\newcommand{\UFL}{UFL}

\newcommand{\CFLext}{Capacitated Facility Location}
\newcommand{\CFL}{CFL}

\newcommand{\VRext}{Vehicle Routing}

\newcommand{\CMDVRext}{Capacitated Multi-Depot Vehicle Routing}
\newcommand{\CMDVR}{CMDVR}

\newcommand{\LRext}{Location Routing}

\newcommand{\LRVDCext}{Capacitated Location Routing}
\newcommand{\LRVDC}{CLR}

\newcommand{\LRVC}{ULR}

\newcommand{\clients}{C}
\newcommand{\depots}{F}
\newcommand{\OPT}{\operatorname{OPT}}
\newcommand{\children}{K_{T'}}

\author{Felipe Carrasco Heine \and Antonia Demleitner \and Jannik Matuschke}

\title{Bifactor Approximation for Location Routing with Vehicle and Facility Capacities}

\begin{document}

\maketitle

\begin{abstract} 
Location Routing is a fundamental planning problem in logistics, in which strategic location decisions on the placement of facilities (depots, distribution centers, warehouses etc.) are taken based on accurate estimates of operational routing costs.  
    We present an approximation algorithm, i.e., an algorithm with proven worst-case guarantees both in terms of running time and solution quality, for the general capacitated version of this problem, in which both vehicles and facilities are capacitated.
    Before, such algorithms were only known for the special case where facilities are uncapacitated or where their capacities can be extended arbitrarily at linear cost. Previously established lower bounds that are known to approximate the optimal solution value well in the uncapacitated case can be off by an arbitrary factor in the general case.
    We show that this issue can be overcome by a bifactor approximation algorithm that may slightly exceed facility capacities by an adjustable, arbitrarily small margin while approximating the optimal cost by a constant factor. 
    In addition to these proven worst-case guarantees, we also assess the practical performance of our algorithm in a comprehensive computational study, showing that the approach allows efficient computation of near-optimal solutions for instance sizes beyond the reach of current state-of-the-art heuristics.
\end{abstract}

\section{Introduction}
\label{sec:intro}

Location decisions play an important role in transport logistics. The placement of distribution centers, warehouses, and depots shapes the topology of logistics networks and has a major influence on overall logistics costs.
When making strategic decisions on the location of such facilities, both the costs for opening and maintaining the facilities as well as the anticipated transportation costs for serving clients from these facilities have to be taken in consideration.
In the classic \FLext{} problem, one of the most widely studied models in location analysis, these latter costs are assumed to be linear, with every client receiving a dedicated connection to a facility.
While this simplification may be appropriate in some situations (in particular outside logistics), in other application contexts, such as the design of local distribution networks, the transportation cost strongly depends on the actual tours that vehicles take to serve the clients. 

It has long been observed that basing location decisions on oversimplified cost models can lead to inferior solutions with costs significantly exceeding those of an optimal solution \citep{Maranzana1964,Webb1968,SalhiRand1989,salhi1999consistency}. This motivates the study of \emph{\LRext},
an integrated approach that combines \FLext{} with a \VRext{} problem:
Given a finite set of possible facility locations, the task is to determine a subset of these facilities that is to be opened and to plan tours originating from the open facilities in order to supply clients and satisfy their demand. These tours have to respect capacity constraints for vehicles (total demand served by a single tour) and facilities (total demand served by a facility). All of this should be done at minimum total cost, which is the sum of opening costs for the facilities and the routing cost, i.e., the total length of all tours.

By combining \FLext{} and \VRext{}, each of which is an $N\!P$-hard problem in its own right, \LRext{} constitutes a computationally challenging problem.
Numerous heuristic and exponential-time exact approaches have been proposed in literature; see the recent surveys by \cite{ProdhonPrins2014} and \cite{DrexlSchneider2015}.
In this paper, we study \emph{approximation algorithms} for \LRext{}, that is, algorithms that come with a proven worst-case guarantee both on their running time and the quality of the produced solution, measured in terms of deviation from the cost of an optimal solution.
In particular, we provide a \emph{bifactor} approximation algorithm, which computes solutions in which the capacity of any facility can be exceeded by at most a small, adjustable factor while the cost is within a constant factor of the optimal solution.
Previously, constant-factor approximation guarantees have only been known for variants of the problem in which facilities are uncapacitated~\citep{RaviSinha2006,HarksKonigMatuschke2013} or for soft-capacitated variants, where facility capacities can be extended arbitrarily at linear cost \citep{chen2009approximation,chen2009cost}.

\subsubsection*{Structure of this paper}
The remainder of this paper is structured as follows. In the rest of \cref{sec:intro}, we give a formal problem definition, followed by a discussion of related work and our contributions. In \cref{sec:bounds}, we discuss two lower bounds on the optimal solution value.
In \cref{sec:algorithm} we present and analyze our approximation algorithm, establishing our main result. \cref{sec:variants} presents heuristic variations of this algorithm, which will be of particular importance for the experiments implemented in \cref{sec:computational}. Finally, \cref{sec:conclusions} contains our conclusions.

\subsection{Formal Problem Definition}
\label{sec:problem}
The \emph{\LRVDCext{} (\LRVDC)} is defined as follows:
As input, we are given a set $\depots$ of possible facility locations, as well as a set $\clients$ of clients who need to be served.
Each facility $w \in \depots$ has an opening cost $f(w) \geq 0$ and a limited capacity $u(w) \in \mathbb{Z}_+$. Every client $v \in \clients$ has an individual demand $d(v) \in \mathbb{Z}_{+}$ that needs to be satisfied. They are served by an unlimited fleet of identical vehicles, each having capacity $\bar{u} \in \mathbb{Z}_+$. 
There further is a metric\footnote{In particular, for every $u, v, w \in V$, it holds that $c(u,v) \leq c(u, w) + c(w, v)$.} distance function $c : V \times V \rightarrow \mathbb{R}_+$, where $V := \depots \cup \clients$, so that $c(v, w)$ denotes the cost for traveling from $v$ to $w$.

A \emph{tour} $T$ consists of a facility $w_T$, a sequence of clients $v^1_T, \dots, v^k_T$, and \emph{service values} $x_T(v)$ for $v \in V$ with $x_T(v) = 0$ for all $v \in \clients \setminus \{v^1_T, \dots, v^k_T\}$.
The \emph{cost} of tour $T$ is $c(T) := c(w_T, v^1_T) + \sum_{i = 1}^{k-1} c(v^i_T, v^{i+1}_T) + c(v^k_T, w_T)$, i.e., the total distance covered by a vehicle starting at $w_T$, visiting the clients in order of the sequence, and returning to $w_T$.

The goal is to decide on a set of facilities $F' \subseteq \depots$ to open and a set of tours~$\mathcal{T}$ such that
\begin{enumerate}
	\item each tour originates from an open facility:\hfill
	$w_T \in F'$ for every $T \in \mathcal{T}$,\label{c:open-fac}
	\item the demand of every client must be served entirely by the tours visiting it:\\
	\hphantom{x} \hfill
	$\sum_{T \in \mathcal{T}} x_T(v) = d(v)$ for all $v \in \clients$,\label{c:client-demand}
	\item the total demand served by any tour is at most the vehicle capacity:\\
	\hphantom{x} \hfill $\sum_{v \in \clients} x_T(v) \leq \bar{u}$ for all $T \in \mathcal{T}$,\label{c:vehicle-cap}
	\item the total demand served from any facility is at most its capacity:\\
	\hphantom{x} \hfill $\sum_{T \in \mathcal{T} : w_T = w} \sum_{v \in \clients} x_T(v) \leq u(w)$ for all $w \in \depots$,\label{c:depot-cap}
\end{enumerate}
minimizing the total cost $\sum_{w \in F'} f(w) + \sum_{T \in \mathcal{T}} c(T)$.

\medskip

\LRVDC{} is a combination of two fundamental optimization problems: \CMDVRext{} and \CFLext{}. Both of which can be seen as special cases of \LRVDC{}.

\medskip

\subsubsection*{{\CMDVRext{}} (\CMDVR{})} The special case where $f \equiv 0$ is known as Capacitated Multi-Depot Vehicle Routing. Note that in this case, it can be assumed that all facilities are opened without loss of generality and thus only the routing aspect is relevant. 

\medskip

\subsubsection*{{\CFLext{}} (\CFL{})} 
In the Capacitated Facility Location problem, we are given the same input as in \LRVDC{}, with the omission  of a vehicle capacity. A feasible solution consists of a set of facilities $F' \subseteq \depots$ to be opened and service values $x(v, w)$ for each client $v \in \clients$ and $w \in F'$ such that $\sum_{w \in F'} x(v, w) = d(v)$ for each $v \in \clients$ and $\sum_{v \in \clients} x(v, w) \leq u(w)$ for all $w \in \depots$. The cost of such a solution is $\sum_{w \in F'} f(w) + \sum_{v \in \clients} c(v, w) x(v, w)$.
This problem is equivalent to an instance of \LRVDC{} in which $\bar{u} = 1$ and $c$ is scaled down by a factor of $\frac{1}{2}$.
Indeed, it is easy to check that for this special case of \LRVDC, there always exists an optimal solution in which every tour corresponds to a round trip visiting only a single client, serving a single unit of demand.
Conversely, because client demands and facility capacities are integer, there exists an optimal solution to the \CFL{} instance in which all $x(v, w)$ are integer.
Interpreting $x(v, w)$ as the number of dedicated round-trips from facility $w$ to client $v$ establishes a one-to-one correspondence among optimal solutions of these special types for the \CFL{} and \LRVDC{} instance, respectively, of same cost.

\medskip

In many application contexts, there are additional desirable properties for location routing solutions that we did not include as hard constraints in our definition of \LRVDC{}. Two noteworthy examples are the single-sourcing and single-tour properties.

\subsubsection*{Single-sourcing/single-tour property}
We say that a solution $(F', \mathcal{T})$ fulfills the \emph{single-sourcing property} if for every client $v \in \clients$ there is a facility $w_v$ such that $w_T = w_v$ for all $T \in \mathcal{T}$ with $x_T(v) > 0$. Furthermore, $(F', \mathcal{T})$ fulfills the \emph{single-tour property}, if for every client $v \in \clients$ there is a unique tour $T \in \mathcal{T}$ with $x_T(v) = d(v)$.

\subsection{Previous Results on Approximating Location Routing and Related Problems}
\label{sec:literature}

We give a brief review of literature on \FLext{}, \VRext{}, and \LRext{}, with a focus on approximation algorithms. An \emph{$\alpha$-approximation algorithm} for an optimization problem is an algorithm that, given an instance of the problem, computes in polynomial time in the size of the input a solution whose cost is at most $\alpha$ times the cost of an optimal solution to that instance; in this context, $\alpha \geq 1$ is called the \emph{approximation factor}. 

The study of approximation algorithms is motivated by the fact that many fundamental optimization problems are NP-hard, leaving little hope for solution methods that are both efficient and exact. Besides their potential to serve as practically usable heuristics, with the additional benefit of an a priori guarantee on the quality of the produced solutions in the worst case, the study of approximation algorithms also yields insights in the inherent complexities of the corresponding optimization problems and the quality of lower bounds on the optimal solution value. For an introduction to the topic see the textbook by~\citet{WilliamsonShmoys2011}.

\subsubsection*{Facility Location}

Facility Location problems have been a major focus in approximation algorithms over the last decades. The approximability of the \UFLext{} problem (\UFL{}), in which there is no upper bound on the demand served by any individual facility, is very well understood in particular.
A large variety of techniques have been shown to be successful in obtaining constant-factor approximations for this setting, including greedy algorithms~\citep{jain2003greedy}, LP-based methods~\citep{ByrkaAardal2010}, and local search~\citep{KorupoluPlaxtonRajaraman2000}, culminating in the currently best-known approximation ratio of $1.48$ \citep{li20131}.

The capacitated version turns out to be much more challenging.
For many years, approximation guarantees could only be proven for a local search with an add/swap/remove neighborhood \citep{KorupoluPlaxtonRajaraman2000}, yielding approximation ratios of $3$ for the case of uniform capacities~\citep{AggarwalLouisBansalGargGuptaGuptaJain2013} and $5$ for arbitrary capacities~\citep{BansalGargGupta2012}.

The difficulty in approximating \CFL{} can be in part be attributed to the fact that the natural LP relaxation for the problem turns out to yield only very weak lower bounds on the optimal solution value in the capacitated setting, with an unbounded gap between optimal integer and fractional solutions.
Only recently, \citet{AnSinghSvensson2017} derived an LP-based constant-factor approximation for \CFL{}, using a considerably more involved LP relaxation.
Moreover, even in the case of uniform capacities, it is NP-hard to decide if a given instance of \CFL{} admits a feasible solution fulfilling the single-sourcing property, a consequence of a straightforward reduction from the Bin Packing problem~\citep{LeviShmoysSwamy2012} that also carries over to \LRVDC{}.

\subsubsection*{Vehicle Routing}
Due to their wide range of applications, the \VRext{} problems have attracted considerable attention in operations research; see the textbooks by \cite{TothVigo2003} and \cite{golden2008vehicle} for an overview.
The fundamental \emph{Traveling Salesperson problem} (TSP), corresponding to the case of a single facility (called \emph{depot} in most of vehicle routing literature) with uncapacitated vehicles, serves as a benchmark both in the design of practical routing algorithms as well as the theoretical study of algorithms in computer science.
In terms of approximation algorithms, the threshold set by the simple and elegant $\nicefrac{3}{2}$-approximation by \cite{christofides1976worst} for TSP has only recently been broken by a considerably more intricate design~\citep{karlin2021slightly}.
For the case of multiple facilities and capacitated vehicles, \citet{LiSmichilevi1990} devised a tour-partitioning technique that can be used to turn any $\rho$-approximation for TSP into a $(2 + \rho)$-approximation.
To the best of our knowledge, no approximation results are known for Multi-Depot Vehicle Routing with capacitated depots.

\subsubsection*{Location Routing}
\citet{RaviSinha2006} paved the way to approximating \LRext{} problems with capacitated vehicles by studying the closely related \emph{Capacitated-Cable Facility Location} (CCFL) problem. This problem corresponds to \LRVDC{} without facility capacities, except that clients are not served by tours but are connected to open facilities via trees of capacitated cables.
\citet{RaviSinha2006} showed that two lower bounds on the optimal cost in a given CCFL instance can be derived by computing solutions to appropriately defined instances of \UFL{} and the Steiner Tree problem, respectively.
Using this insight, they derived a $(\rho_{\text{\UFL{}}} + \rho_{\text{ST}})$-approximation for CCFL, where $\rho_{\text{\UFL{}}}$ and $\rho_{\text{ST}}$ denote approximation factors of algorithms for \UFL{} and Steiner Tree, respectively.
Based on their framework, \citet{HarksKonigMatuschke2013} showed how to obtain a $4.32$-approximation for \LRext{} with capacitated vehicles but uncapacitated facilities.
Our algorithm, discussed below, combines this framework with an LP-rounding scheme to balance the load at individual facilities to obtain a bifactor approximation for \LRVDC{}.

\citeauthor{RaviSinha2006}'s (\citeyear{RaviSinha2006}) article concludes with the question whether it is possible to obtain constant-factor approximation results for variants of CCFL with additional constraints, mentioning facility capacities as a notable example of practical importance. \cite{chen2009approximation,chen2009cost} give a partial answer to this question by studying two \emph{soft-capacitated} variants of CCFL and Location Routing, respectively:
The \emph{Soft-capacitated Facility Location and Cable Installation} (SC-FLCI) problem corresponds to CCFL, in which arbitrarily many copies of each facility can be opened, each copy costing $f(w)$ and being able to serve a demand of $u(w)$.
The \emph{Access Network Design} (AND) problem, motivated by a problem in telecommunication network design, corresponds to location routing in which both vehicles and facilities are uncapacitated, but in which an additional cost $a(w)$ is incurred for each unit of demand served by each facility $w \in F$, representing a linear cost for installing sufficient capacity at each facility. Note that in both problems, no hard upper bound on the total amount of demand served by any facility exists---hence the name ``soft-capacitated''.
\citeauthor{chen2009approximation} provide a 19.84-approximation for SC-FLCI~(\citeyear{chen2009approximation}) and a 12-approximaiton for AND~(\citeyear{chen2009cost}), both based on a primal-dual approach.

\subsection{Our Contribution}
In this paper, we study approximation algorithms for \LRVDC{} with uniform vehicle capacities and arbitrary facility capacities.
We start by pointing out that, contrary to the case without facility capacities, the natural adaptation of the combined tree/facility-location lower bound by \cite{RaviSinha2006} to the case with facility capacities can be arbitrarily far off from the value of an optimal solution. 

Motivated by this observation and the fact that computing feasible solutions fufilling the single-sourcing property is $N\!P$-hard, we turn our attention to bifactor approximation:
Our main result is an algorithm that given a \LRext{} instance and a parameter $\varepsilon \in (0, 1]$, computes in polynomial time a solution in which every facility $w$ receives a load of at most $u(w) + \varepsilon \bar{u}$, where $u(w)$ is the capacity of the facility and $\bar{u}$ is the uniform vehicle capacity.
The cost of this solution is bounded by $4 + \nicefrac{2\alpha}{\varepsilon}$ times that of an optimal solution to the original instance, where $\alpha$ is the approximation guarantee of an algorithm used to compute the \CFL{} lower bound.\footnote{The current best known approximation factors for \CFL{} are $\alpha = 5$ for the general case~\citep{BansalGargGupta2012} and $\alpha = 3$ for the case of uniform facility capacities~\citep{AggarwalLouisBansalGargGuptaGuptaJain2013}, respectively.}
For the special case of Capacitated Multi-Depot Vehicle Routing, we can further assume $\alpha = 1$, obtaining a factor of $4 + \nicefrac{2}{\varepsilon}$.
The solutions produced by our algorithms furthermore fulfill the single-sourcing and the stronger single-tour property (i.e., every client is served entirely by a single visit of a vehicle), as long as all client demands are bounded by $\varepsilon \bar{u}$. 

Our algorithm uses an adaptation of the framework by \citet{RaviSinha2006} to cluster the clients into groups. It then assigns these groups to open facilities via a rounding scheme for a linear program that balances the load at individual facilities. The approximation guarantee is proven using the aforementioned tree/facility-location lower bounds, and as a consequence, solutions with cost close to these lower bounds exist and can be computed when relaxing facility capacities by an arbitrarily small factor. 

We remark that, in many application contexts, typical facility capacities are at least an order of magnitude larger than vehicle capacities, and so the violation of facility capacities caused by our algorithm is relatively small even when $\varepsilon$ is set to $1$.
It is also important to point out that, differently from the soft-capacitated variants of the problem discussed earlier, the total demand served by any facility $w \in F$ is strictly limited the a priori upper bound $u(w) + \varepsilon \bar{u}$.
This makes our results applicable to settings in which small extensions of facility capacities are admissible but larger extensions (which are assumed to be possible in the soft-capacitated setting) are impossible, e.g., due to physical limitations.

Complementing our theoretical results, we devise several heuristic modifications to the algorithm that improve its practical performance. In particular, by replacing the aforementioned load-balancing linear program by an integer program of moderate size, we show how our framework can be used to obtain solutions that strictly respect the original facility capacities while still achieving a comparably low cost.
We analyze the empirical performance of the algorithm with and without heuristic improvements in an extensive computational study on different sets of benchmark instances from literature and additional newly generated instances.
Clearly outperforming its theoretical worst-case guarantee, the algorithm is capable of computing near-optimal solutions that either slightly exceed facility capacities (when using the original polynomial-time variant) or strictly respect these capacities (when using the integer program) in a fraction of the running time necessary for exact or other heuristic approaches.

\section{Tree and Facility-Location Lower Bounds}
\label{sec:bounds}

In the following we discuss two combinatorial lower bounds introduced by~\citet{RaviSinha2006} for the capacitated cable facility-location problem in the slightly adapted form used  by \citet{HarksKonigMatuschke2013} for \LRext{} with capacitated vehicles and uncapacitated facilities.
Both bounds are straightforward to adapt to the case with facility capacities.
Throughout the rest of the paper, we will let $\OPT$ denote the cost of an optimal solution for a \LRVDC{} instance~(the instance in question will always be clear from the context).

\subsection{Spanning Tree Lower Bound}
The first bound is based on a minimum spanning tree---it ignores both vehicle and facility capacities and can be applied to \LRVDC{} without adaptation; see Lemma~2 by \citet{HarksKonigMatuschke2013} for a proof.

\begin{lemma} \label{lem:mst_lb}
For a given \LRVDC{} instance, consider the graph $G = (V \cup \{r\}, E)$ with $E := \{ \{r,w\}: w \in \depots \} \cup \{ \{v, w\} : v \in \clients, w \in \depots\} \cup \{ \{v, v'\} : v, v' \in \clients, v \neq w\}$. Define costs $c'(r, w) = 0$, $c'(v, w) = c(v, w) + \frac{1}{2} f(w)$ for all $v \in \clients$, $w \in \depots$, and $c'(v, v') = c(v, v')$ for all other $\{v, v'\} \in E$. Let $T'$ be a minimum spanning tree in $G$ with respect to weights $c'$. Then, $c'(T') \leq \OPT$.
\end{lemma}

\subsection{Capacitated Facility Location Lower Bound}
The second lower bound is based on facility location. For \LRVC{} the bound uses an instance of \UFLext{}. It is straightforward to see that by using a corresponding instance of Capacitated Facility Location instead, one obtains a lower bound for \LRVDC{}.

\begin{lemma} \label{lem:cfl_lb}
For a given \LRVDC{} instance, consider the following \CFL{} instance: clients, facilities, demands, opening costs, and facility capacities remain the same as in the \LRVDC{} instance, but the distances are set to $\tilde{c} := \nicefrac{2c}{\bar{u}}$.
Let $(\tilde{F}, \tilde{x})$ be an optimal solution to this \CFL{} instance.
Then $\sum_{w \in \tilde{F}} \big(f(w) + \sum_{v \in \clients} \tilde{c}(v, w) \tilde{x}(v, w)\big) \leq \OPT$.
\end{lemma}

\begin{proof}
  Let $(F', \mathcal{T})$ be an optimal solution to the \LRVDC{} instance.
  For $v \in \clients$ and $w \in F'$ define $x'(v, w) = \sum_{T \in \mathcal{T} : w_T = w} x_T(v)$.
  Note that $(F', x')$ is a feasible solution to the \CFL{} instance. Moreover,
  \begin{align*}
    \sum_{v \in \clients} \sum_{w \in F'} \tilde{c}(v, w) x'(v, w) & = \sum_{v \in \clients} \sum_{w \in F'} \sum_{T \in \mathcal{T} : w_T = w} \frac{2}{\bar{u}} c(v, w) x_T(v)\\
    & = \sum_{T \in \mathcal{T}} \sum_{v \in \clients}  2c(v, w_T) \frac{x_T(v)}{\bar{u}}\\
    & \leq \sum_{T \in \mathcal{T}} c(T)  \sum_{v \in \clients} \frac{x_T(v)}{\bar{u}} \leq \sum_{T \in \mathcal{T}} c(T),  \end{align*}
    where the penultimate inequality follows from the fact that any tour $T$ containing a client $v$ can be decomposed into a $w_T$-$v$-path and a $v$-$w_T$-path, each of which has length at most $c(v, w_T)$ by triangle inequality.
  Hence, the total cost of the \CFL{} solution $(F', x')$ is at most the cost of the \LRVDC{} solution $(F', \mathcal{T})$.
\end{proof}

\subsection{Approximation Gap for Lower Bounds in the Capacitated Setting}
In the setting where facilities are uncapacitated, the approximation result of \cite{HarksKonigMatuschke2013} implies that, for any instance, the maximum of the two lower bounds is within a constant factor of the value of an optimal solution. 
The following lemma reveals that this does no longer hold in the capacitated setting. 

\begin{lemma}\label{lem:lower_bounds_not_enough}
  For any $n \in \mathbb{N}$, there exists an instance of \LRVDC{} with $n$ clients and two facilities such that $\OPT \geq (n - 1) \max \{L', \tilde{L}\}$, where $L'$ and $\tilde{L}$ are the values of the lower bounds described in \cref{lem:mst_lb,lem:cfl_lb}, respectively.
\end{lemma}
\begin{proof}
Consider a \LRVDC{} instance with $n = |\clients|$ clients with unit demands $d \equiv 1$. There are two possible facility locations: $\depots = \{w_1, w_2\}$ with $f(w_1) = f(w_2) = 0$ and $u(w_1) = u(w_2) = \bar{u} = n-1$. All clients are located at the same position as $w_1$, which has a distance of $1$ to $w_2$. That is, $d(v, w_1) = 0$ and $d(v, w_2) = d(w_1, w_2) = 1$ for all $v \in \clients$.

Observe that for this instance the value of the MST lower bound as described in \cref{lem:mst_lb} is $0$ because facility opening cost are $0$ and every client is located at distance $0$ to a facility.
Observe further that the value of the \CFL{} lower bound as described in \cref{lem:cfl_lb} is $\nicefrac{2}{\bar{u}} = \nicefrac{2}{(n-1)}$, as all but one client can be served by facility $w_1$ at cost $0$.
Note, however, that any feasible solution to the \LRVDC{} instance needs to contain a tour connecting at least one client to facility $w_2$ and that the cost of such a tour is at least $2$. Thus the cost of an optimal solution for the constructed instance is at least $n-1$ times the larger of the two lower bounds.
\end{proof}

In the next section, however, we will show that a solution within a constant factor of the two lower bounds can still be obtained when slightly relaxing the facility capacities. The following corollary is an implication of the analysis presented in \cref{sec:algorithm}.

\begin{corollary}
  For any $\varepsilon \in (0, 1]$ and any instance of \LRVDC{}, there exists a solution in which constraints \eqref{c:open-fac}-\eqref{c:vehicle-cap} are fulfilled, the load at any facility $w \in \depots$ is no more than $u(w) + \varepsilon \bar{u}$, and the total cost is no more than $4L' + \frac{2}{\varepsilon} \tilde{L}$, where $L'$ and $\tilde{L}$ are the values of the lower bounds described in \cref{lem:mst_lb,lem:cfl_lb}, respectively.
\end{corollary}

\section{Approximation Algorithm}
\label{sec:algorithm}

In the following, we present a bifactor approximation algorithm for \LRVDC{}. The algorithm and the accompanying analysis prove the following result:

\begin{theorem}\label{thm:approximation}
There is an algorithm that, given $\varepsilon \in (0, 1]$, computes in polynomial time a solution to a given instance of {\LRVDC} such that conditions \eqref{c:open-fac}-\eqref{c:vehicle-cap} are fulfilled, the load at any facility $w \in \depots$ is no more than $u(w) + \varepsilon \bar{u}$, and the total cost is no more than $(4 + \frac{2\alpha}{\varepsilon})\OPT$, where $\alpha \geq 1$ is the approximation factor of an algorithm for {\CFL}.
\end{theorem}

Throughout this section let $\varepsilon \in (0, 1]$ be the parameter given in the input of the algorithm and define $\bar{\bar{u}} := \varepsilon \bar{u}$.

\subsubsection*{Overview}
The algorithm consists of three steps.
In the first step, a minimum spanning tree for the instance described in \cref{lem:mst_lb} is partitioned into clusters such that each cluster contains clients with a total demand of at most $\bar{\bar{u}}$ and every cluster with demand less than $\nicefrac{\bar{\bar{u}}}{2}$ contains a facility.
The clustering technique is based on a procedure for relieving overloaded subtrees by \cite{AlpertKahngMandoiuZelikovsky2003} and \cite{RaviSinha2006}.
In the second step, the clusters are assigned to open facilities via a rounding procedure for an assignment LP. Using the fact that most clusters have large aggregated demand, it is shown that solutions to the {\CFL} instance described in \cref{lem:cfl_lb} induce feasible solutions of bounded cost for this LP.
The rounding procedure might allocate one additional cluster per facility, thus resulting in an additional demand of at most $\bar{\bar{u}}$ at any open facility.
Finally, each cluster is converted into a tour by adding an edge to its assigned facility an using the classic doubling-and-shortcutting technique for TSP.

\subsection{Step 1: Clustering}
\label{sec:relieving}

The first step of the algorithm consists of a clustering procedure that partitions the set of clients into appropriate clusters and associated trees connecting the nodes within each cluster.

\begin{figure}[t]
	\centering
	\begin{tikzpicture}[scale=0.9]
    \path (0, -0.7) node {$T'$};
		\path node[facNode] (r) {}
			+ (-1.5, 1) node[normalNode] {} edge[normalEdge] (r)
			++ (1, 1.25) node[normalNode] (v0) {} edge[normalEdge] (r)
			+ (-1.5, 0.5) node[normalNode] (v1) {} edge[normalEdge] (v0) 
			+ (0, 1) node[normalNode] (v2) {} edge[normalEdge] (v0) 
			+ (2, 1) node[normalNode] (v4) {} edge[normalEdge] (v0);
		\path (v1) 
			+(-0.5, 1) node[normalNode] {} edge[normalEdge] (v1);
		\path (v2) 
			+(0, 1) node[normalNode] {} edge[normalEdge] (v2);
		\path (v4)
			+(0.5, 0.5) node[normalNode] {} edge[normalEdge] (v4);
		
		\node[draw, single arrow, minimum height=8mm, minimum width=0.75mm, single arrow head extend=.5mm] at (4.5,1.5) {};
		
		\begin{scope}[xshift=8cm]
			\path (0, -0.7) node {$T'$};
    		\path node[facNode] (r) {}
    			+ (-1.5, 1) node[labeledNode] (ve) {} edge[normalEdge] (r)
    			++ (1, 1.25) node[labeledNode] (v0) {} edge[normalEdge] (r)
    			+ (-1.5, 0.5) node[labeledNode] (v1) {} edge[normalEdge] (v0) 
    			+ (0, 1) node[labeledNode] (v2) {} edge[normalEdge] (v0) 
    			+ (1, 1.9) node[normalNode] (v3) {} edge[dashedEdge] (v0)   
    			+ (2, 1) node[labeledNode] (v4) {} edge[normalEdge] (v0);
    		\path (v1) 
    			+(-0.35, 1) node[normalNode] {} edge[normalEdge] (v1)
    			+(0.35, 1) node[normalNode] {} edge[dashedEdge] (v1);
    		\path (v2) 
    			+(-0.5, 0.5) node[normalNode] {} edge[dashedEdge] (v2)
    			+(0, 1) node[normalNode] {} edge[normalEdge] (v2);
    		\path (v4) 
    			+(-0.5, 0.5) node[normalNode] {} edge[dashedEdge] (v4)
    			+(0.5, 0.5) node[normalNode] {} edge[normalEdge] (v4);
    		\path (ve)
    		    +(-0.35, 0.75) node[normalNode] {} edge[dashedEdge] (ve)
    			+(0.35, 0.75) node[normalNode] {} edge[dashedEdge] (ve);
		\end{scope}
	\end{tikzpicture}
	\caption{Preprocessing of the tree. If a client $v$ does not occur as a leaf, or if $d(v) > \bar{\bar{u}}$, it is replaced by a dummy node whose demand is equal to $0$. For each such client, $\ell := \lceil d(v) / \bar{\bar{u}} \rceil$ nodes are added at distance $0$ from the original, each with a demand of $d(v)/\ell$. In the figure, hollow circles represent dummy nodes, filled circles represent clients with positive demand, and the square represents a facility. In the modified tree, each of these new nodes is connected to the dummy node corresponding to $v$ (the corresponding edges are depicted as dashed lines).}
	\label{fig:preproc}
\end{figure}
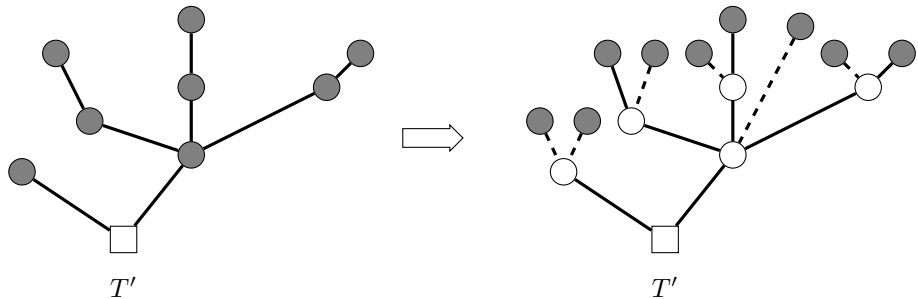

\subsubsection*{Preprocessing of the tree}
The clustering procedure starts with a minimum spanning tree $T'$ for the instance $(G, c')$ described in \cref{lem:mst_lb}.
Note that we can assume $T'$ to contain the edges $\{r, w\}$ for all $w \in \depots$ without loss of generality (as those edges have cost $0$).
We further modify the tree such that such that clients only occur at leafs and all clients have demand at most $\bar{\bar{u}}$.
This can be achieved by splitting each client $v$ in $\lceil d(v) / \bar{\bar{u}} \rceil$ nodes, distributing the demand uniformly among them, and by introducing an additional \emph{dummy node} at the location of each client that takes the role of the original, possibly internal, node of the client in the tree. The nodes resulting from splitting the client, which have positive demand, are attached to this dummy node as leafs. See \cref{fig:preproc} for an illustration. Dummy nodes will be removed when constructing the tours at the end of the algorithm, so they do not occur in the final solution.
Note that neither of these operations increases the cost of the tree and that clients with demand at most $\bar{\bar{u}}$ remain represented by a single client.

\subsubsection*{Notation}
For $v, w \in V(T')$ we say that $w$ is a \emph{descendant} of $v$ in $T'$ if $v$ is on the unique $r$-$w$-path in $T'$.
We say that $w$ is a \emph{child} of $v$ in $T'$, if $v$ directly precedes $w$ on the unique $r$-$w$-path in $T'$.
We use the notation $T'[v]$ to denote the subtree of $T'$ containing $v$ and all its descendants. We denote the set of children of $v$ in $T'$ by $\children(v)$. 
We further let $d_{T'}(v) := \sum_{v' \in V(T'[v]) \cap \clients} d(v')$ denote the total demand in the subtree of $v$.

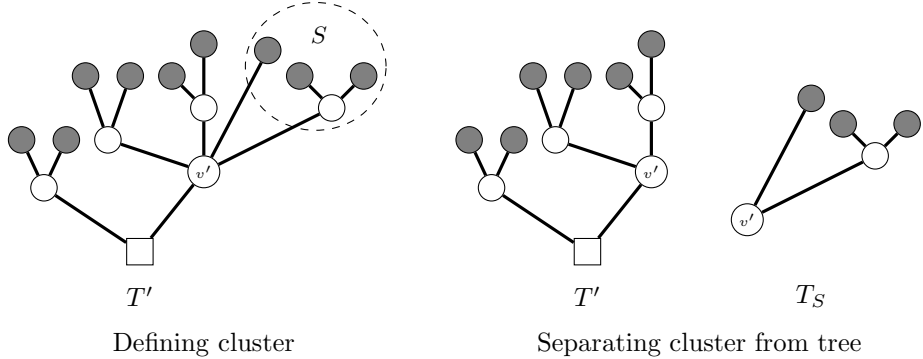
\begin{figure}[t]
	\centering
	\begin{tikzpicture}[scale=0.85]
    \path (0, -0.7) node {$T'$}
      +(1, -0.75) node {Defining cluster};
      
        \path node[facNode] (r) {}
			+ (-1.5, 1) node[labeledNode] (ve) {} edge[normalEdge] (r)
			++ (1, 1.25) node[labeledNode] (v0) {\tiny $v'$} edge[normalEdge] (r)
			+ (-1.5, 0.5) node[labeledNode] (v1) {} edge[normalEdge] (v0) 
			+ (0, 1) node[labeledNode] (v2) {} edge[normalEdge] (v0) 
			+ (1, 1.9) node[normalNode] (v3) {} edge[normalEdge] (v0)   
			+ (2, 1) node[labeledNode] (v4) {} edge[normalEdge] (v0);
		\path (v1) 
			+(-0.35, 1) node[normalNode] {} edge[normalEdge] (v1)
			+(0.35, 1) node[normalNode] {} edge[normalEdge] (v1);
		\path (v2) 
			+(-0.5, 0.5) node[normalNode] {} edge[normalEdge] (v2)
			+(0, 1) node[normalNode] {} edge[normalEdge] (v2);
		\path (v4) 
			+(-0.5, 0.5) node[normalNode] {} edge[normalEdge] (v4)
			+(0.5, 0.5) node[normalNode] {} edge[normalEdge] (v4);
		\path (ve)
		    +(-0.35, 0.75) node[normalNode] {} edge[normalEdge] (ve)
			+(0.35, 0.75) node[normalNode] {} edge[normalEdge] (ve);

		\draw[dashed, thin] (v3) +(0.75, -0.25) ellipse (1.1cm and 1cm);
		\path (v3) ++(0.8, 0.25) node {\small $S$};

		\begin{scope}[xshift=7cm]
			\path (0, -0.7) node {$T'$}
				+(1.75, -0.75) node {Separating cluster from tree}
				++(3.5, 0) node {$T_S$};
			\path node[facNode] (r) {}
        		+ (-1.5, 1) node[labeledNode] (ve) {} edge[normalEdge] (r)
    			++ (1, 1.25) node[labeledNode] (v0) {\tiny $v'$} edge[normalEdge] (r)
    			+ (-1.5, 0.5) node[labeledNode] (v1) {} edge[normalEdge] (v0) 
    			+ (1.5, -0.75) node[labeledNode] (v0a) {\tiny $v'$}
    			+ (0, 1) node[labeledNode] (v2) {} edge[normalEdge] (v0);
    		\path (v0a)
    			+ (1, 1.9) node[normalNode] (v3) {} edge[normalEdge] (v0a)   
    			+ (2, 1) node[labeledNode] (v4) {} edge[normalEdge] (v0a);
    		\path (v1) 
    			+(-0.35, 1) node[normalNode] {} edge[normalEdge] (v1)
			    +(0.35, 1) node[normalNode] {} edge[normalEdge] (v1);
    		\path (v2) 
    			+(-0.5, 0.5) node[normalNode] {} edge[normalEdge] (v2)
    			+(0, 1) node[normalNode] {} edge[normalEdge] (v2);
    		\path (v4) 
    			+(-0.5, 0.5) node[normalNode] {} edge[normalEdge] (v4)
    			+(0.5, 0.5) node[normalNode] {} edge[normalEdge] (v4);
    		\path (ve)
    		    +(-0.35, 0.75) node[normalNode] {} edge[normalEdge] (ve)
    			+(0.35, 0.75) node[normalNode] {} edge[normalEdge] (ve);
		\end{scope}
	
	\end{tikzpicture}
	\caption{A step of the clustering procedure. 
	As in \cref{fig:preproc}, hollow circles represent dummy nodes (with no demand) and filled circles represent clients (with positive demand).
	The procedure finds a node $v'$ with $d_{T'}(v') > \bar{\bar{u}}$ but $d_{T'}(v) \leq \bar{\bar{u}}$ for all $v \in \children(v')$. It greedily constructs a set $S$ such that $d(S) \geq \nicefrac{\bar{\bar{u}}}{2}$. The set $S$ is added to $\mathcal{S}$, and the corresponding tree $T_S$ is removed from $T'$.}
	\label{fig:relieving}
\end{figure}

\subsubsection*{Clustering}
The clustering procedure creates a family $\mathcal{S}$ (initially empty) of client sets, together with corresponding trees $T_S$ for each $S \in \mathcal{S}$. 
To this end, the procedure repeatedly identifies a node $v' \in V(T') \setminus \{r\}$ such that $d_{T'}(v') > \bar{\bar{u}}$ but $d_{T'}(v) \leq \bar{\bar{u}}$ for all $v \in \children(v')$.
Then a subset $L \subseteq \children(v')$ of the children of $v'$ is greedily selected so that $\nicefrac{\bar{\bar{u}}}{2} \leq \sum_{v \in L} d_{T'}(v) \leq \bar{\bar{u}}$ (note that this is possible because we assume $d(v) \leq \bar{\bar{u}}$ for all $v \in \clients$).
The client set $S := \bigcup_{v \in L} V(T'[v]) \cap \clients$ is added as a new element of $\mathcal{S}$, associated with the corresponding tree $T_S := \bigcup_{v \in L} T'[v] \cup \{v, v'\}$ consisting of the subtrees induced by the nodes in $L$ and the edges connecting these subtrees to $v'$.
The edges of $T_S$ and its nodes except for $v'$ are removed from $T'$ (see \cref{fig:relieving}).
Once $d_{T'}(v') \leq \bar{\bar{u}}$ for all nodes $v' \in V(T')$ in the remaining tree, the procedure identifies a set of facilities $F_1$ containing all facilities $w$ for which $V(T'[w]) \cap \clients \neq \emptyset$.
Each such set is added as an additional cluster to $\mathcal{S}$, together with the corresponding tree $T'[w]$.
The algorithm returns the clustering $\mathcal{S}$, the corresponding trees $T_S$ for $S \in \mathcal{S}$ and the set of facilities $F_1$.
A pseudo-code listing of the clustering procedure is given in \cref{alg:clustering} and an illustration is given in \cref{fig:relieving}.

\begin{algorithm}[t]
\DontPrintSemicolon
\caption{Clustering} \label{alg:clustering}
Let $T'$ be a minimum spanning tree in the graph $G'$ with weights $c'$.\;
Initialize $\mathcal{S} := \emptyset$.\;
\While{$\textup{there is } v \in V(T') \setminus \{r\} \textup{ with } d_{T'}(v) > \bar{\bar{u}}$}{
Let $v' \in V(T') \setminus \{r\}$ such that $d_{T'}(v') > \bar{\bar{u}}$ but $d_{T'}(v) \leq \bar{\bar{u}}$ for all $v \in \children(v')$.\;
Let $L \subseteq \children(v')$ be such that $\nicefrac{\bar{\bar{u}}}{2} \leq \sum_{v \in L} d_{T'}(v) \leq \bar{\bar{u}}$.\;
Add the set $S := \bigcup_{v \in L} V(T'[v]) \cap \clients$ to $\mathcal{S}$ and let $T_S := \bigcup_{v \in L} T'[v] \cup \{v, v'\}$.\;
Remove all edges of $T_S$ and all nodes of $V(T_S) \setminus \{v'\}$ from $T'$.\;
}
Let $F_1 := \{w \in \depots : V(T'[w]) \cap \clients \neq \emptyset\}$.\;
\For{$w \in F_1$}{
Add the set $S := V(T'[w]) \cap \clients$ to $\mathcal{S}$ and let $T_S := T[w]$.
}
\Return{$(\mathcal{S}, \{T_S : S \in \mathcal{S}\}, F_1)$}
\end{algorithm}

Note that while it is possible for the node sets $V(T_S)$ and $V(T_{S'})$ for $S, S' \in \mathcal{S}$ with $S \neq S'$ to overlap at a dummy node or a facility, the client sets $S$ and $S'$ are disjoint (recall that clients only occur at leafs of the tree $T'$) and the same is true for the edge sets of the trees $T_S$ and $T_{S'}$.
Hence $\mathcal{S}$ indeed comprises a partition of the clients and the trees $T_S$ partition the original tree $T'$ (with the exception of the edges $\{r, w\}$ for $w \in \depots$ which do not occur in any tree).
These observations and additional properties of the clustering following immediately from its construction in the algorithm are summarized below.

\begin{lemma}\label{lem:clustering}
  \cref{alg:clustering} computes in polynomial time a partition $\mathcal{S}$ of the clients together with a tree $T_S$ for each $S \in \mathcal{S}$ such that:
  \begin{itemize}
    \item $T_S \subseteq T'$ and $T_S \cap T_{S'} = \emptyset$ for all $S, S' \in \mathcal{S}$ with $S \neq S'$,
    \item $S \subseteq V(T_S)$ for all $S \in \mathcal{S}$,
    \item $\sum_{v \in S} d(v) \leq \bar{\bar{u}}$ for all $S \in \mathcal{S}$.
  \end{itemize}
  Moreover, defining $\mathcal{S}' := \{S \in \mathcal{S} : \sum_{v \in S \cap \clients} d(v) < \nicefrac{\bar{\bar{u}}}{2}\}$, there is a unique facility $w_S \in V(T_S)$ for every $S \in \mathcal{S}'$, and $w_{S} \neq w_{S'}$ for $S \neq S'$.
\end{lemma}

\subsection{Step 2: Assignment}
\label{sec:assign_subp}

The second step of the algorithm constructs an assignment of the clusters constructed in the first step to a set of open facilities. 

\subsubsection*{The assignment LP}
To this end, the algorithm first computes an approximate solution to the \CFL{} instance described in \cref{lem:cfl_lb}.
Let us denote the set of facilities opened in this solution by $F_2$ and the corresponding assignment by $\tilde{x}$.
Recall further the clustering $\mathcal{S}$ and the set of facilities $F_1$ returned by the first step of the algorithm, and define $F' := F_1 \cup F_2$.
For $S \in \mathcal{S}$ and $w \in F$, we define $c(S, w) := \min_{v \in V(T_S)} c(v, w)$ and $d(S) := \sum_{v \in S} d(v)$.
Consider the following assignment LP that assigns the demand of each cluster produced by the first step of the algorithm to a facility in~$F'$:
\begin{equation}\label{LP_trans}
\begin{array}{rrcll}
\min & \displaystyle \sum_{S\in\mathcal{S}} \sum_{w\in F'} \frac{c(S,w)}{d(S)} x(S,w)  \\
\textrm{s.t.} & \displaystyle  \sum_{S\in\mathcal{S}} x(S,w) & \leq & u(w) & \forall w\in F' \\
& \displaystyle  \sum_{w\in F} x(S,w) & = & d(S) & \forall S\in\mathcal{S} \\
& x(S,w) & \geq & 0 & \forall S \in \mathcal{S},w\in F'
\end{array}
\end{equation}
In the proof of the following lemma, we show that the {\CFL} solution induces a feasible solution to~\eqref{LP_trans} of bounded cost.

\begin{lemma}\label{lem:assignment-lp}
  The linear program \eqref{LP_trans} has a feasible solution with objective function value at most $\frac{1}{\varepsilon}\sum_{v \in \clients} \sum_{w \in F'} \tilde{c}(v, w) \tilde{x}(v, w) + \sum_{S \in \mathcal{S}'} c(T_S)$.
\end{lemma}

\begin{proof}
To prove the lemma, we first construct an instance of a network flow problem in an undirected graph $\bar{G} = (\bar{V}, \bar{E})$ on the node set $\bar{V} := V(T') \setminus \{r\}$ (thus, $V'$ contains all clients and facilities, as well as potential dummy nodes introduced in the preprocessing step).
The edge set $\bar{E} = \bar{E}_1 \cup \bar{E}_2$ consists of the edges $\bar{E}_1 := \{\{v, w\} : v \in \clients, w \in \depots, \tilde{x}(v, w) > 0\}$ representing the support of the {\CFL} assignment and $\bar{E}_2 := \bigcup_{S \in \mathcal{S}'} T_S$, i.e., the union of all trees $T_S$ corresponding to clusters $S$ with low demand $d(S) < \nicefrac{\bar{\bar{u}}}{2}$.

For $e = \{v, w\} \in \bar{E}_1$ we set the capacity $u(e) := \tilde{x}(v, w)$ and for $e \in \bar{E}_2$ we set the capacity $u(e) := d(S)$, where $S \in \mathcal{S}'$ is the unique cluster with $e \in T_S$.

For $w \in F'$ let $S_w$ denote the unique cluster $S \in \mathcal{S}'$ with $w_S = w$, if it exists, and let $S_w = \emptyset$ otherwise.
We will consider each facility $w \in F'$ as a source with supply $s(w) := u(w) - d(S_w)$. We will consider any client $v \in \bar{\clients} := \clients \setminus \bigcup_{S \in \mathcal{S}'} S$ that is not in a cluster of low demand as a sink  with demand $r(v) := d(v)$.

The following claim establishes the existence of a flow in $\bar{G}$ that saturates all these demands while not exceeding any edge capacity or any supply at a facility. To make this formal, let $\mathcal{P}_{vw}$ for any $v, w \in \bar{V}$ denote the set of $v$-$w$-paths in the graph $\bar{G}$ and let $\mathcal{P} := \bigcup_{v \in \bar{\clients}, w \in F'} \mathcal{P}_{wv}$.

\begin{claim}\label{claim:flow}
    There exists a flow $y \in \mathbb{R}_+^\mathcal{P}$
    such that $\sum_{P \in \mathcal{P} : e \in P} y(P) \leq u(e)$ for all $e \in E$,
    $\sum_{v \in \bar{\clients}} \sum_{P \in \mathcal{P}_{vw}} y(P) \leq s(w)$ for all $w \in F'$, and $\sum_{w \in F'} \sum_{P \in \mathcal{P}_{vw}} y(P) \geq r(v)$ for all $v \in \bar{\clients}$.
\end{claim}

We postpone the proof of the claim and first show how it implies the lemma.
Note that by choice of the capacities $u$, the cost of the flow $y$ in terms of $c$ is bounded by
\begin{align}
  \sum_{e \in \bar{E}} c(e) \!\!\!\!\! \sum_{P \in \mathcal{P} : e \in P} \!\!\!\!\!\! y(P)  & \leq \sum_{v \in \clients} \sum_{w \in F'} c(v, w) \tilde{x}(v, w) + \sum_{S \in \mathcal{S}'} \sum_{\{v, w\} \in T_S} \!\!\!\!\! c(v, w) d(S).\label{eq:flow-cost}
\end{align}

Let $\bar{y}(v, w) := \sum_{P \in \mathcal{P}_{vw}} y(P)$ be the amount of flow that is sent from $w \in F'$ to $v \in \bar{\clients}$ in the flow $y$.
Consider the vector $x$ defined by $x(S, w_S) := d(S)$ for all $S \in \mathcal{S}'$ and $x(S, w) := \sum_{v \in S} \bar{y}(v, w)$ for all $S \in \mathcal{S} \setminus \mathcal{S}'$ and $w \in F'$, with $x(S, w) = 0$ for all other pairs $S \in \mathcal{S}'$ and $w \neq w_S$.
Note that this vector is a feasible solution to \eqref{LP_trans} because $\sum_{S \in \mathcal{S}} x(S, w) = d(S_w) + \sum_{v \in \bar{\clients}} \bar{y}(v, w) \leq d(S_w) + s(w) = u(w)$ for all $w \in F'$, $\sum_{w \in F'} x(S, w) = d(S)$ for all $S \in \mathcal{S}'$, and
$\sum_{w \in F'} x(S, w) = \sum_{w \in F'} \bar{y}(v, w) = r(v) = d(v)$ for all $v \in \bar{\clients}$.
Moreover,
\begin{align}
  \sum_{S\in\mathcal{S}} \sum_{w\in F'} \frac{c(S,w)}{d(S)} x(S,w) & \leq \sum_{S\in \mathcal{S} \setminus \mathcal{S}'} \sum_{w\in F'} \sum_{v \in S} \frac{c(v,w)}{d(S)} \bar{y}(v,w) \notag\\
    & \leq \frac{2}{\bar{\bar{u}}} \sum_{S\in \mathcal{S} \setminus \mathcal{S}'} \sum_{w\in F'} \sum_{v \in S} c(v, w) \sum_{P \in \mathcal{P}_{vw}} \!\! y(P) \notag\\
    & \leq \frac{2}{\bar{\bar{u}}} \sum_{w \in F'} \sum_{v \in {\clients}} \sum_{P \in \mathcal{P}_{vw}} \sum_{e \in P} c(e) y(P)\notag \\
    & = \frac{2}{\bar{\bar{u}}} \sum_{e \in \bar{E}} \sum_{P \in \mathcal{P} : e \in P} c(e) y(P),\label{eq:solution-x-cost}
\end{align}
where the first inequality follows from $c(S, w_S) = 0$ for all $S \in \mathcal{S}'$, the second follows from the fact that $d(S) \geq \bar{\bar{u}}/2$ for all $S \in \mathcal{S} \setminus \mathcal{S}'$, and the third follows from the fact that $c$ is a metric and hence $c(v, w) \leq \sum_{e \in P} c(e)$ for any $P \in \mathcal{P}_{vw}$.
Combining \eqref{eq:flow-cost} and \eqref{eq:solution-x-cost} yields 
\begin{align*}
\sum_{S\in\mathcal{S}} \sum_{w\in F'} \frac{c(S,w)}{d(S)} x(S,w) & \leq \frac{1}{\varepsilon}\sum_{v \in \clients} \sum_{w \in F'} \tilde{c}(v, w) \tilde{x}(v, w) + \sum_{S \in \mathcal{S}'} \sum_{\{v, w\} \in T_S} c(v, w)
\end{align*}
using the definition of $\tilde{c}$ and the fact that $\nicefrac{d(S)}{\bar{\bar{u}}} \leq \nicefrac{1}{2}$ for $S \in \mathcal{S}'$.
We thus established that \cref{claim:flow} implies \cref{lem:assignment-lp}. The proof of the former is given below.
\end{proof}

\begin{proof}[Proof of \cref{claim:flow}.]
By the max-flow/min-cut theorem, a flow satisfying the demands $r$ while respecting supply limits $s$ and capacities $u$ exists if and only if
$$\sum_{v \in A \cap \bar{\clients}} r(v) - \sum_{w \in A \cap F'} s(w) \leq \sum_{e \in \delta(A)} u(e)$$ for any node set $A \subseteq V'$, where $\delta(A)$ is the cut induced by $A$, i.e., the set of edges $e \in \bar{E}$ with exactly one endpoint in $A$ and the other in $V' \setminus A$. It thus suffices to show that the above inequality holds for any $A \subseteq V'$.

Let $A \subseteq V'$ and note that
\begin{align*}
  \sum_{v \in A \cap \bar{\clients}} r(v) - \sum_{w \in A \cap F'} s(w) & = \sum_{v \in A \cap \clients} d(v) - \sum_{S \in \mathcal{S}'} \sum_{v \in A \cap S} d(v) - \sum_{w \in A \cap F'} u(w) - d(S_w)\\
  & = \sum_{v \in A \cap \clients} d(v) - \sum_{w \in A \cap F'} u(w) + \sum_{w \in A \cap F_1} d(S_w \setminus A)\\
  & \leq \sum_{v \in A \cap \clients} \sum_{w \in F' \setminus A} \tilde{x}(v, w) + \sum_{w \in A \cap F_1} d(S_w \setminus A),
\end{align*}
where the last inequality follows from the fact that $\tilde{x}$ is a feasible assignment for the {\CFL} instance with open facilities $F_2 \subseteq F'$.
Finally, note that
$\sum_{v \in A \cap \clients} \sum_{w \in F' \setminus A} \tilde{x}(v, w) \leq \sum_{e \in \delta(A) \cap E_1} u(e)$
and that $\delta(A) \cap \bar{E}_2$ contains at least one edge $e \in T_{S_w}$ for every $w \in A \cap F_1$ for which $S_w \setminus A \neq \emptyset$ (see \cref{fig:max-flow-min-cut} for an illustration). Because any such edge $e \in T_{S_w} \subseteq \bar{E}_2$ has capacity $u(e) = d(S_w)$, we obtain $\sum_{w \in A \cap F_1} d(S_w \setminus A) \leq \sum_{e \in \delta(A) \cap \bar{E}_2} u(e)$, proving the claim.
\end{proof}

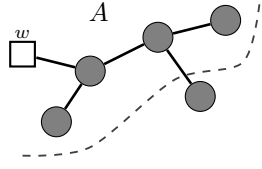
\begin{figure}[t]
\centering
\begin{tikzpicture}[
scale=.45,
customer/.style={circle, draw=black, fill=gray, scale = 1.2},
facility/.style={rectangle, draw=black, thick, scale = 1.4},
]
  
  \node[rectangle] at (2.25,3.25) {$A$};
  \node[rectangle] at (0,2.6) {\scriptsize $w$};
  \node[facility] (o1) at (0,2) {};
  \node[customer] (c1) at (2,1.5) {};
  \node[customer] (c2) at (4,2.5) {};
  \node[customer] (c3) at (6,3) {};
  \node[customer] (c4) at (1,0) {};
  \node[customer] (c5) at (5.25,0.75) {};
    
  \draw[line width=1pt, color=black] (o1) to (c1);
  \draw[line width=1pt, color=black] (c1) to (c2);
  \draw[line width=1pt, color=black] (c2) to (c3);
  \draw[line width=1pt, color=black] (c1) to (c4);
  \draw[line width=1pt, color=black] (c2) to (c5);
  
  \draw [line width=0.7pt, black!75, dashed] plot [smooth] coordinates {(0,-1) (2,-0.75) (4.5,1.25) (6.5,1.75) (7,3.5)};
 
\end{tikzpicture}
\caption{The cut induced by set $A$ in the proof of \cref{claim:flow}. If $w \in A \cap F_1$ but $S_w \setminus A \neq \emptyset$, then at least one edge of $T_{S_w}$ crosses the cut.}
\label{fig:max-flow-min-cut}
\end{figure}

\subsubsection*{Rounding the assignment}
The algorithm computes an optimal extreme point solution $x$ to LP \eqref{LP_trans}. The following rounding procedure transforms $x$ into a new solution $x'$ with $x'(S, w) \in \{0, d(S)\}$ for all $S \in \mathcal{S}$ and $w \in F'$.
Starting with $x' := x$, the rounding procedure maintains a helper graph $G_{x'} = (\mathcal{S} \cup F', E_{x'})$ with edge set $E_{x'} = \{\{S, w\} : S \in \mathcal{S}, w \in F', 0 < x'(S, w) < d(S)\}$, which is updated along with the solution.
As long as $E_{x'} \neq \emptyset$, it iteratively applies the following modification to $x'$ (which intuitively is a flow augmentation along a path in the helper graph): 
Let $w, w' \in F'$ be two facilities such that both $w$ and $w'$ are incident to exactly one edge in $G_x$ and such that there is a $w$-$w'$-path $P$ in $G_x$ (it is argued that in the proof of \cref{lem:flow-rounding} that these always exist while the graph is non-empty).
Let $w_0 := w, S_1, w_1, \dots, w_{k-1}, S_k, w_k := w'$ denote the nodes along path $P$ in order of traversal from $w$ to $w'$ and let $I := \{(S_i, w_{i-1}) : i \in \{1, \dots, k\}$ and $I' := \{(S_i, w_{i}) : i \in \{1, \dots, k\}$.
Without loss of generality (by swapping the roles of $w$ and $w'$ if necessary), we can assume that $\sum_{(S, w) \in I} c(S, w) \leq \sum_{(S, w) \in I'} c(S, w)$. Let $\Delta := \min \{d(S) - x(S, w) : (S, w) \in I\} \cup \{x(S, w) : (S, w) \in I'\}$ and update $x$ by increasing $x(S, w)$ by $\Delta$ for all $(S, w) \in I$ and decreasing $x(S, w)$ by $\Delta$ for all $(S, w) \in I'$.

\begin{algorithm}[t]
\DontPrintSemicolon
\caption{Assignment} \label{alg:assignment}
Let $x$ be an optimal extreme point solution to \eqref{LP_trans}. Initialize $x' := x$.\;
\While{$\textup{there is } S \in \mathcal{S} \textup{ and } w \in F' \textup{ with } 0 < x'(S, w) < d(S)$}{
Let $E_{x'} := \{\{S, w\} : S \in \mathcal{S}, w \in F', 0 < x'(S, w) < d(S)\}$.\;
Let $w, w' \in F'$ such that both $w$ and $w'$ have degree $1$ and there exists a $w$-$w'$-path $P = (w_0, S_1, w_1, \dots, S_k, w_k)$ in the graph $G_{x'} = (\mathcal{S} \cup F', E_{x'})$. \quad (note: $w_0 = w,\ w_k = w'$)\;
Let $I := \{(S_i, w_{i-1}) : i \in \{1, \dots, k\}\}$ and $I' := \{(S_i, w_{i}) : i \in \{1, \dots, k\}\}$.\;
\If{$\sum_{(S, w) \in I} c(S, w) \leq \sum_{(S, w) \in I'} c(S, w)$}{
Let $\Delta := \min \{d(S) - x'(S, w) : (S, w) \in I\} \cup \{x'(S, w) : (S, w) \in I'\}$.\; 
Let $x'(S, w) := x'(S, w) + \Delta$ for all $(S, w) \in I$.\;
Let $x'(S, w) := x'(S, w) - \Delta$ for all $(S, w) \in I'$.\;
}
\Else{
Let $\Delta := \min \{d(S) - x'(S, w) : (S, w) \in I'\} \cup \{x'(S, w) : (S, w) \in I\}$.\; 
Let $x'(S, w) := x'(S, w) + \Delta$ for all $(S, w) \in I'$.\;
Let $x'(S, w) := x'(S, w) - \Delta$ for all $(S, w) \in I$.\;
}
}
\Return{$x'$}
\end{algorithm}

\bigskip

The following lemma shows that that this procedure terminates after a linear number of iterations and results in a solution $x'$ whose cost is at most that of $x$ and which allocates at most a demand of $u + \bar{\bar{u}}$ to each facility.\footnote{We remark that procedures yielding the same guarantees as those in Lemma~\ref{lem:flow-rounding} are implied by numerous rounding procedures for more general linear programs, such as the classic $2$-approximation for the generalized assignment problem~\citep{shmoys1993approximation} or 
more recent approximation results for single-source unsplittable flow~\citep{morell2020single}. 
However, the procedure outlined here is specifically tailored to the special case resulting from LP~\eqref{LP_trans}, yielding a simpler and more efficient algorithm.}

\begin{lemma}
  \label{lem:flow-rounding}
  The assignment procedure (\cref{alg:assignment}) computes in polynomial time (after at most $|\mathcal{S}| + |F'| - 1$ iterations) a vector $x'$ with $x'(S, w) \in \{0, d(S)\}$ for all $S \in \mathcal{S}$ and all $w \in F'$ such that 
  \begin{itemize}
    \item $\sum_{S \in \mathcal{S}} \sum_{w \in F'} \frac{c(S, w)}{d(S)} x'(S, w) \leq \sum_{S \in \mathcal{S}} \sum_{w \in F'} \frac{c(S, w)}{d(S)} x(S, w)$,
    \item $\sum_{w \in F'} x'(S, w) = d(S)$ for all $S \in \mathcal{S}$, and
    \item $\sum_{S \in \mathcal{S}} x'(S, w) \leq \bar{\bar{u}} + \sum_{S \in \mathcal{S}} x(S, w)$ for all $w \in F'$.
  \end{itemize}
\end{lemma}

\begin{proof}
	Note that any modification applied to $x'$ in the algorithm does not increase the cost and keeps $\sum_{w \in F'} x'(S, w)$ invariant for all $S \in \mathcal{S}$ (because the selected path always starts and ends at a facility).
	Hence if the algorithm terminates, the first and second condition are automatically met.

	Initially, because $x$ is an extreme point solution to a transportation problem, the graph $G_{x'}$ is a forest before the start of the first iteration.
	During the course of the algorithm, once a variable $x'(S, w)$ is set to either $0$ or $d(S)$ its value is never changed again (because the corresponding edge cannot occur in the path $P$ anymore).
	Hence, in any iteration, the set of edges in $E_{x'}$ is a subset of the edges from the previous iteration, and $G_{x'}$ is a forest throughout the run of the algorithm.
  
  We further observe that no cluster node $S \in \mathcal{S}$ can have degree $1$ in $G_{x'}$, as $0 < x'(S, w) < d(S)$ for some $w$ implies that there must be at least one other $w'$ with $0 < x'(S, w) < d(S)$.
  Hence, any non-singleton connected component of $G_{x'}$ is a tree whose leafs are elements of $F'$.
  Hence, as long as $E_{x'}$ is non-empty, there are indeed $w, w' \in F'$ with degree $1$ that are connected by a $w$-$w'$-path in $G_{x'}$.
  
  Furthermore, note that $\sum_{S \in \mathcal{S}} x'(S, w)$ for some $w \in F'$ can only change while $w$ is a leaf in $G_{x'}$ (because the changes in the variables $x'$ cancel out for $w$ when it occurs in the interior of the path $P$).
  Consider the first iteration when $w \in F'$ becomes a leaf (if any such iteration exists) and note that $\sum_{S \in \mathcal{S}} x'(S, w) = \sum_{S \in \mathcal{S}} x(S, w) \leq u(w)$ at the begin of this iteration.
  Let $S' \in \mathcal{S}$ be the unique cluster such that $\{S', w\} \in E_{x'}$.
  Because for any $S \neq S'$, the variable $x'(S, w)$ will not be changed anymore and $x'(S', w)$ is never increased beyond $d(S') \leq \bar{\bar{u}}$, we obtain $\sum_{S \in \mathcal{S}} x'(S, w) \leq \sum_{S \in \mathcal{S}} x(S, w) +\bar{\bar{u}}$ throughout the algorithm.
  
	Finally, by choice of $\Delta$, there is at least one variable $x'(S, w)$ in each iteration whose value is set to either $0$ or $d(S)$.
	Hence the number of edges in $E_{x'}$ decreases by at least one in each iteration.
	Because $E_{x'}$ is a forest, it contains at most $|\mathcal{S}| + |F'| - 1$ edges initially and the algorithm terminates after $|\mathcal{S}| + |F'| - 1$ iterations.
	Note further that $E_{x'} = \emptyset$ implies $x'(S, w) \in \{0, d(S)\}$ for all $S \in \mathcal{S}$ and $w \in F'$.
\end{proof}

\subsection{Step 3: Constructing tours}
In the final step of the algorithm, the trees covering each cluster are connected to the facilities as indicated by the rounded assignment $x'$. For each cluster, the resulting tree is turned into a tour via the classic edge-doubling-and-shortcutting procedure.
For every $S \in \mathcal{S}$, let $w_S \in F'$ be the unique facility with $x'(S, w_S) = d(S)$ and let $v \in V(T_S)$ be such that $c(v, w_S) = c(S, w_S)$.
Define $\bar{T}_S := T_S \cup \{v, w_S\}$.
Note that is a connected graph $\bar{T}_S$ spanning (a superset of) the nodes in $S \cup \{w_S\}$ and hence a tour visiting the clients in $S$ and facility $w_S$ of length at most $2\big(c(T_S) + c(S, w_S)\big)$ can be computed using the classic double tree algorithm:
(1) double the edges in $\bar{T}_S$, (2) compute a Eulerian tour on the doubled tree, (3) shortcut the tour by skipping any nodes noting $(S \cap \clients) \cup \{w_S\}$ and any repeated visits of a node.
The algorithm returns the resulting tours together with the set of facilities $F'$ to be opened.

\subsubsection*{Analysis}
In the constructed solution, the total demand on all tours based at facility $w \in F'$ is $$\sum_{\stackrel[w_S = w]{}{S \in \mathcal{S}}} d(S) = \sum_{S \in \mathcal{S}} x'(S, w) \leq \bar{\bar{u}} + \sum_{S \in \mathcal{S}} x'(S, w) \leq u(w) + \varepsilon \bar{u}$$ by \cref{lem:flow-rounding}.
Furthermore, the cost of the tours constructed by the algorithm is bounded by
$\sum_{S \in \mathcal{S}} 2 (c(T_S) + c(S, w_s))$. 
Note that $\sum_{S \in \mathcal{S}} c(T_S) \leq c(T')$ because the trees $T_S$ for $S \in \mathcal{S}$ are edge-disjoint by \cref{lem:clustering}.
Moreover,
\begin{align*}
  \sum_{S \in \mathcal{S}} c(S, w_S) = \sum_{S \in \mathcal{S}} \frac{c(S, w)}{d(S)} x'(S, w) \leq \frac{1}{\varepsilon}\sum_{v \in \clients} \sum_{w \in F'} \tilde{c}(v, w) \tilde{x}(v, w) + \sum_{S \in \mathcal{S}'} c(T_S)
\end{align*}
by \cref{lem:assignment-lp,lem:flow-rounding}.
We conclude that the total cost of the constructed tours is bounded by $4c(T') + \frac{2}{\varepsilon}\sum_{v \in \clients} \sum_{w \in F'} \tilde{c}(v, w) \tilde{x}(v, w)$.
Finally, note that every facility $w \in F_1$ is incident to at least one edge $e \in T'$ other than $\{r, w\}$ and hence
$c(T') + \sum_{w \in F_1} f(w) \leq c'(T)$.
Combining these observations, we can bound the total cost of the produced solution by
\begin{align*}
  \sum_{w \in F'} f(w) + \sum_{T \in \mathcal{T}} c(T) & \leq 4c'(T') + \sum_{w \in F_2} f(w) + \frac{2}{\varepsilon} \sum_{v \in \clients} \sum_{w \in F'} \tilde{c}(v, w) \tilde{x}(v, w)\\
  & \leq \left(4 + \frac{2\alpha}{\varepsilon}\right) \operatorname{OPT},
\end{align*}
where $\alpha$ is the approximation factor for the {\CFL} solution computed in step~$2$.

We further note that clients $v \in \clients$ whose demand in the original input was bounded by $\varepsilon \bar{u}$ were not split during the preprocessing and assigned to a single tour in the algorithm.
As each of the three steps can be executed in polynomial time, this concludes the analysis of the algorithm and the proof of \cref{thm:approximation}.

\section{Algorithm Variants and Heuristic Improvements}
\label{sec:variants}

In this section, we discuss some heuristic modifications and variants of the algorithm that, while not yielding better worst-case guarantees, help to improve the practical performance of the algorithm. We will assess the effect of these modifications in \cref{sec:computational}.

\subsection{Accounting for Facilities Opened in Step~1} \label{sec:free_fac}
The algorithm presented in the preceding section computes two sets of facilities, $F_1$ in Step~1 and $F_2$ in Step~2, and the final solution opens all facilities in $F_1 \cup F_2$. As the set of facilities in $F_1$ is already known at the beginning of Step~2, the cost of those facilities can be reduced to $0$ in the \CFL{} instance that is solved in that step. This encourages the algorithm to make use of the facilities opened in Step~1 instead of opening additional facilities, heuristically reducing the opening cost of the produced solution.

\subsection{Reducing \CFL{} Instance Sizes} \label{sec:reduce_size}

Step~2 of the algorithm presented in \cref{sec:algorithm} requires approximately solving a \CFL{} instance of the same size as the original \LRVDC{} instance in order to determine the set $F_2$. While polynomial-time constant factor approximation algorithms for \CFL{} exist, these local-search based algorithms still require significant computational effort and dominate the running time of our algorithm as instance sizes increase.
As an alternative approach, one can solve the smaller \CFL{} instance obtained from merging the clusters computed in Step~1 of the algorithm, i.e., similarly to the assignment LP~\eqref{LP_trans}, we define $\tilde{c}(S, w) := \min_{v \in V(T_S)} \tilde{c}(v, w)$ and $d(S) := \sum_{v \in S} d(v)$ for $S \in \mathcal{S}$ and $w \in F$. 
Not only does this modification reduce the number of clients in the \CFL{} instance roughly by a factor of $\bar{\bar{u}}$, the resulting instance also intuitively anticipates the use of the computed set $F_2$ in the construction of the final solution, where facilities are indeed connected to clusters rather then individual clients. It is thus to be expected that this modification will in fact reduce not only computation time but also solution cost.\footnote{We remark that the resulting \CFL{} instance does not necessarily fulfill the triangle inequality, however, most \CFL{} algorithms can still be applied as heuristics in this case, as this property is only used in their analysis.}

\subsection{Integer Program for Combined \CFL{} and Assignment} \label{sec:cfl_ip}
Instead of approximately solving a \CFL{} instance to determine a set of facilities to be opened and then applying the rounding procedure to the solution of \eqref{LP_trans} in order to assign clusters to open facilities, these two steps can be directly addressed in a single step using the following integer program:
\begin{equation}\label{ip_cfl}
\begin{array}{rrcll}
\min & \displaystyle \sum_{S\in\mathcal{S}} \sum_{w\in F} 2c(S,w) y(S,w) & + & \sum_{w \in \depots} f(w) z(w)  \\
\textrm{s.t.} & \displaystyle  \sum_{S\in\mathcal{S}} d(S) y(S,w) & \leq & u(w)z(w) & \forall w\in \depots \\
& \displaystyle  \sum_{w\in F} y(S,w) & = & 1 & \forall S\in\mathcal{S} \\
& y(S,w) & \in & \{0, 1\} & \forall S \in \mathcal{S}, w\in \depots\\[3pt]
& z(w) & \in & \{0, 1\} & \forall w\in \depots
\end{array}
\end{equation}
Here, $c(S,w)$ is defined in the same way as in \cref{sec:assign_subp}.
The factor $2$ in the objective function coefficients for the variables $y$ accounts for the fact that a vehicle not only needs to visit a cluster, but will also travel back to the corresponding facility afterwards, resulting in a better estimation of the routing costs.

A solution to IP \eqref{ip_cfl} not only yields a set of facilities to be opened (those $w \in \depots$ with $z(w) = 1$) but also an assignment of client clusters to open facilities, which can be used in Step~3 of the algorithm.
Even more, such an assignment also leads to a feasible \LRVDC{} solution, as capacity constraints are satisfied due to the first constrained of \eqref{ip_cfl}.

As with the preceding modification, this approach benefits from the fact that the number of clusters is significantly smaller than the number of clients in the original input instance. Thus solving the IP is tractable at least for moderately sized \LRVDC{} instance.
We remark, however, that contrary to LP~\eqref{LP_trans}, which is always feasible when the original \LRVDC{} instance is feasible, there is no such guarantee for IP~\eqref{ip_cfl}. To address this issue, we suggest the following hierarchical approach: Whenever IP~\eqref{ip_cfl} turns out to be infeasible, find the smallest $\gamma > 1$ such that the 
IP~\eqref{ip_cfl} with the right hand-side of the first constraint replaced by 
$\gamma u(w)z(w)$ is feasible and proceed with an optimal solution to this modified~IP.

\subsection{Post-optimization of Tours via LKH}\label{sec:lkh}
While typical \LRVDC{} instances may consist of a large number of clients, individual tours usually only contain a small number of clients due to the limited vehicle capacity.
For such small sets of clients, optimal or almost-optimal tours can be found efficiently using LKH, the implementation of the Lin-Kernighan heuristic by \cite{helsgaun2000effective}, in Step~3 of the algorithm.
\section{Computational Study}
\label{sec:computational}

In this section, we assess the practical performance of our algorithm in a computational study on a great variety of instances. We investigate the quality of produced solutions and the scalability of the approach.

\subsection{Implementation Details}\label{sec:impdet}

In our experiments, we tested four variants of our algorithm, differing in the implementation of Step~2 and the use of LKH post-optimization:

\medskip

\begin{itemize}
    \item $A_{\text{LS}}^{\text{DTS}}$: the algorithm described in \cref{sec:algorithm}, but using the improvements described in \cref{sec:free_fac,sec:reduce_size}, i.e., the \CFL{} instance in Step~2 is constructed with modified facility costs and clustered clients; $3$-approximation local search procedure by \citet{AggarwalLouisBansalGargGuptaGuptaJain2013} is used to solve the \CFL{} instance in Step~2
    \item $A_{\text{IP}}^{\text{DTS}}$: the algorithm described in \cref{sec:algorithm}, but using the assignment IP \eqref{ip_cfl} to replace Step~2, as described in \cref{sec:cfl_ip}
    \item $A_{\text{LS}}^{\text{LKH}}$: same as $A_{\text{LS}}^{\text{DTS}}$, but the tours are optimized using LKH as described in \cref{sec:lkh}
    \item $A_{\text{IP}}^{\text{LKH}}$: same as $A_{\text{IP}}^{\text{DTS}}$, but the tours are optimized using LKH as described in \cref{sec:lkh}
\end{itemize}

\medskip

Parameter $\varepsilon$ was fixed to $1$. Preliminary experiments showed that, not surprisingly, smaller values of $\varepsilon$ reduce capacity violations but increase solution cost quite drastically. As the IP-based algorithms $A_{\text{IP}}^{\text{DTS}}$ and $A_{\text{IP}}^{\text{LKH}}$ also produce feasible solutions in almost all cases, complete experiments with smaller values of $\varepsilon$ were omitted.

The algorithm was implemented in Python~3.7.3, using Gurobi~8.1.1 to solve subproblems~\eqref{LP_trans} and \eqref{ip_cfl}, respectively. To avoid excessive execution times on larger instances, we use Gurobi's \emph{TimeLimit} option, which returns the best found solution within the specified time interval.
We also preemptively terminate the Gurobi's optimization process if the incumbent solution remains unchanged for a specified period. Similarly, we interrupt the local search procedure for \CFL{} if the whole procedure exceeds a specified time limit, or if the current solution cannot be improved within a certain time limit. The exact values of these time limits were set according to instance size and are specified in \cref{sec:test_instances}. 

Extensive use of Python's \textit{networkx} library was made, although avoiding methods that rely on randomness or arbitrariness in order to obtain reproducible results. We also used the LKH implementation available in the \textit{elkai} library.
The experiments were run on an AMD Ryzen 5 PRO 2400G processor at 3.6 GHz, with 16 GB RAM.

We point out that the main goal of our computational experiments is providing a proof-of-concept, showing that the approach used to obtain the approximation result can be the basis of a capable and scalable heuristic. We do not necessarily aim at obtaining competitive results, which is why we refrained from tuning the implementation for speed or using more time-efficient methods for solving the \CFL{} instances in Step~2 of the algorithm.

\subsection{Test Instances} \label{sec:test_instances}

We tested the algorithms both on benchmark instances previously used in literature as well as on newly generated random instances.

\subsubsection*{Benchmark instances from literature}
Many instance sets for \LRVDC{} have been published in literature over the previous decades.
Among them, instance sets provided by \citet{BarretoFerreiraPaixaoSousasantos2007}, \citet{PrinsProdhonWolflercalvo2006}, and \citet{TuzunBurke1999} have emerged as a standard benchmark for \LRVDC{} algorithms. 
The first one consists of $13$ instances with both facility and vehicle capacity. The number of clients ranges from $21$ to $150$, and the amount of possible facility locations ranges from $5$ to $14$. \citeauthor{PrinsProdhonWolflercalvo2006}'s collection contains $30$ instances in which the amount of clients ranges from $20$ to $200$, and the number of potential facility locations ranges from $5$ to $10$. Again, both vehicles and facilities have limited capacity. Finally, \citeauthor{TuzunBurke1999}'s set is the largest, with $36$ instances. Vehicles have limited capacity, while the capacity of each facility is equal to the total demand of all clients, making them uncapacitated in practice. The amount of clients of these instances lies between $100$ and $200$, while the number of possible locations lies between $10$ and $20$. 
We compare our algorithms against the best known solution (BKS) values for these instances, as reported by \citet{SchneiderDrexl2017}. We will denote this set of $79$ instances by \textbf{PBT}. 

In addition, we included the collection created by \citet{SchneiderLoffler2019}, which contains slightly larger instances. In the $202$ instances included in this set,  the number of facilities ranges from $5$ to $30$, and the number of clients lies between $100$ and $600$, following a similar structure as the instances by \citet{PrinsProdhonWolflercalvo2006}. This set not only allows us to test our method on bigger instances than the previous standard, but \citeauthor{SchneiderLoffler2019} also provide best known solutions which we can compare our results with. We will denote this instance collection by \textbf{S\&L}.

\subsubsection*{Additional randomly generated instances}
To assess the scalability of our approach, we created additional randomly generated instances with the number of clients ranging between $50$ and $10000$. In generating these instances, we slightly adapted the procedures specified by \citet{TuzunBurke1999} and by \citet{SchneiderLoffler2019}. The main components of this process were the following:

\begin{itemize}
    \item The number of clients of an instance is specified by a value from the following range: 50, 100, 150, 200, 300, 400, 500, 600, 700, 800, 900, 1000, 2500, 5000, or 10000. For a given instance, we will refer to this value as the \textit{instance size} $n$. The number of potential facility locations in any instance of size $n$ is simply given by $\nicefrac{n}{20}$, except when the number of clients is either $50$ or $150$, in which case the number of possible facility locations equals $5$ and $10$, respectively.
    \item A number of conglomerates is specified from the set $\{0, 3, 5\}$. This number is used to randomly generate the distance metric described below.
    \item The vehicle capacity of an instance is specified from the set $\{70, 150, 300\}$.
    \item The uniform facility capacity is specified from the set  $\{400, 600, 1200\}$. 
    \item For each instance, one of three possible ranges for the facility costs is chosen: $[2, 4]$, $[200, 400]$, or $[20000, 40000]$. The cost of each facility is drawn uniformly at random from the chosen range. This leads to instances in which the opening costs are either considerably lower than the routing costs, comparable to them, or considerably higher.
    
    \item  To determine the (Euclidean) distances, clients and facilities are randomly assigned a point in $[0, 1000]^2$ as follows: If the number of conglomerates specified in the input is $0$, all clients and facilities are placed uniformly and independently at random in this region.
    If the number of conglomerates is $3$ or $5$, the region $[0, 1000]^2$ is divided in a $3 \times 3$-grid of $9$ cells, of which $3$ or $5$, respectively, are chosen as \emph{conglomeration areas}. Then, $80\%$ of the clients and $80\%$ of the facilities are randomly distributed across these conglomeration areas, and each one of them is assigned uniformly at random a coordinate in the corresponding cell. The other $20\%$ of clients and facilities are then equally distributed across the remaining cells, and each one of them is assigned uniformly at random a coordinate in the corresponding cell as well.
    
    \item The demand of each client is drawn independently and uniformly from the interval $[10, 20]$. 
\end{itemize}

Each instance created is named $n$-$k$-$abc$, where $n$ stands for the number of clients, $k$ for the number of client conglomerates, and $a, b, c \in \{\text{s}, \text{m}, \text{l}\}$, referring to small, medium, or large vehicle capacity, facility costs, or facility capacity, respectively.
According to this scheme, we considered all possible combinations for any given value of $n \leq 1000$ and create $3^4 = 81$ instances for each such value. For the larger instances (size $2500$ to $10000$), we used Taguchi's orthogonal array design to cover possible combinations of the parameters in a fractional factorial design with $9$ combinations that we tested for $n \in \{2500, 5000, 10000\}$; see \cref{tab:exp} for the detailed design. 

In total, we thus obtained $999$ instances. We partition this collection into four sets according to instance size: \textbf{S} for instances with $50 \leq n \leq 200$, \textbf{M} for instances with $300 \leq n \leq 600$, \textbf{L} for instances with $700 \leq n \leq 1000$, and \textbf{XL} for instances with $2500 \leq n \leq 10000$.

\begin{table}[t]
\centering
\begin{tabular}{|c c c c|} 
    \hline
    \# of clusters & facility cost & vehicle capacity & facility capacity \\ 
    \hline\hline
    \multirow{3}{*}{0} & s & s & s \\
    & m & m & m \\
    & l & l & l \\
    \hline
    \multirow{3}{*}{3} & s & m & l \\  
    & m & l & s \\ 
    & l & s & m \\ 
    \hline
    \multirow{3}{*}{5} & s & l & m \\ 
    & m & s & l \\ 
    & l & m & s \\
    \hline
\end{tabular}
\caption{Experimental design for XL instances.}
\label{tab:exp}
\end{table}

\subsubsection*{Lower bounds for additional instances} \label{sec:best_lb}
Since no previous solutions are available for the newly generated instances, we use the lower bounds described \cref{lem:mst_lb,lem:cfl_lb} to measure the quality of the produced solutions. 
During the clustering step of the algorithm, the corresponding MST needs to be computed, providing us with the lower bound specified in \cref{lem:mst_lb}. 
Additionally, we compute solutions to the \CFL{} instance described in \cref{lem:cfl_lb} for every instance, resulting in the second lower bound.
For instances of size S to L, we compute exact solutions to these \CFL{} instances using the standard integer programming formulation for \CFL{}.
In the case of the XL instances, due to memory constraints, we used a sparsification approach to reduce the size of the integer programming formulation\footnote{This procedure consists in dividing the area in smaller cells, and letting clients only be directly served by facilities in neighboring cells. Each cell additionally contains an artificial hub that can re-distribute the supply of facilities that are located further away. On instances for which the exact solution to \CFL{} could be computed, this sparsification led to a mild deterioration of at most $2.4\%$ in the computed lower bounds.}.
For each instance, the greater of these two lower bounds (\CFL{} or MST) is chosen as a reference. The gap is then computed as $\Delta_{\text{LB}} = \nicefrac{(\text{ALG} - \text{LB})}{\text{LB}}$, where $\text{ALG}$ and $\text{LB}$ denote the cost of the solution produced by the algorithm and the lower bound, respectively.

\subsubsection*{Time limits by instance size}
As mentioned in \cref{sec:impdet}, time limits were set to avoid excessive computation times for large instances. The following time limits were set depending on the instance size:

\begin{itemize}
\item For improving the incumbent solution in either the local search or while solving the integer program in Step~2, time limits of $60$, $90$, and $120$ minutes were set for instances in sets M, L, and XL, respectively.
\item For the entire solution process of either the local search or solving the integer program in Step~2, time limits of $180$, $270$, and $360$ minutes were set for instances in sets M, L, and XL, respectively.
\item No time limits were set for the instances belonging to sets PBT, S\&L, and S.
\end{itemize}

\subsection{Results}

In this section, we present and analyze the results of our computational experiments. 
We focus on the results obtained by $A_{\text{LS}}^{\text{DTS}}$ and $A_{\text{IP}}^{\text{LKH}}$. While the former is the closest representation of the bifactor approximation described in \cref{sec:algorithm}, the latter results in the best solution quality. We discuss the differences between all four variants at the end.

\subsubsection*{Solution quality on benchmark instances}

For the benchmark instances from literature, we compare the solutions obtained by our algorithms to the best known solution for the respective instance. We report the gap $\Delta_{\text{BKS}} = \nicefrac{(\text{ALG} - \text{BKS})}{\text{BKS}}$, where $\text{ALG}$ and $\text{BKS}$ denote the cost of the solution produced by the algorithm and the best known solution, respectively.
For algorithm $A_{\text{LS}}^{\text{DTS}}$, we observed average gaps of $18.64\%$ ($22.49\%$ when restricted to those cases where the solution of the algorithm does not violate any facility capacities) on the instance set PBT and $7.74\%$ ($13.97\%$ when restricted to feasible solutions) on the instance set S\&L. Algorihtm $A_{\text{IP}}^{\text{LKH}}$ computed feasible solutions on for all instances of both sets, with average gaps of $11.69\%$ for PBT and $5.23\%$ for S\&L.
See \cref{fig:bksgap} for a detailed depiction of the results.
We observe that the algorithm thus computes solutions much closer to best-known (in many cases proven optimal) solutions than guaranteed by our worst-case analysis in~\cref{sec:algorithm}. 
For $14$ of the S\&L instances, $A_{\text{IP}}^{\text{LKH}}$ even computes new best known solutions.

\begin{figure}[t]
    \centering
    \includegraphics[width=0.7\linewidth]{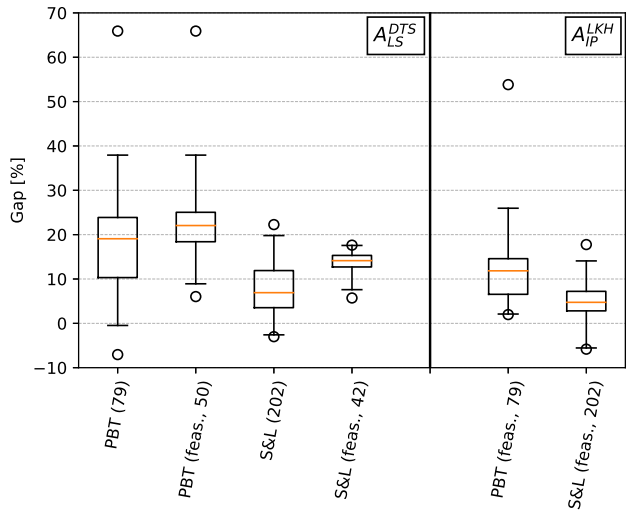}
    \caption{Gap to best known solution for instances in the benchmark set.
    For algorithm $A_{\text{LS}}^{\text{DTS}}$, the results for all instances in the respective set and the results for only those instances in the set where the algorithm produced a feasible solution (i.e., no violation of capacity limits; indicated by `feas.') is shown. Numbers in parentheses indicate the respective number of instances. Algorithm $A_{\text{IP}}^{\text{LKH}}$ produced feasible solutions for all instances in the benchmark set.} 
    \label{fig:bksgap}
\end{figure}

\subsubsection*{Solution quality on newly generated instances}

Using the lower bounds described in \cref{sec:best_lb}, we were able to measure the performance of our algorithm on our newly created instances. We observe that the distribution of the lower bound gap is not significantly affected by instance size, with global average gaps of $63.88\%$ ($81.80\%$ when restricted to feasible solutions) for $A_{\text{LS}}^{\text{DTS}}$ and $56.44\%$ ($55.61\%$) for $A_{\text{IP}}^{\text{LKH}}$ (see \cref{fig:lbgap}). These gaps are all considerably lower than the worst-case guarantee of the algorithm. Additionally, $A_{\text{IP}}^{\text{LKH}}$ results in several feasible near-optimal solutions, with gaps as low as $2.15\%$.

Among the infeasible solutions to the instances in set XL, two outliers were left out of \cref{fig:lbgap}: One solution with a gap of $388.25\%$ computed by $A_{\text{LS}}^{\text{DTS}}$ and one with a gap of $416.38\%$ computed by $A_{\text{IP}}^{\text{LKH}}$. In both cases, the respective procedures in Step~2 of the algorihtm (LS or integer program, respectively) were interrupted while the incumbent solutions were still frequently being improved. Nevertheless, both values are still below the guarantee provided by the theoretical analysis.

As we will discuss below, most instances in the set XL resulted in a premature termination of the algorithm due to the time limit.
However, with the exception of the two outliers mentioned above, this had no significant influence on the lower bound gaps, which exhibit practically the same distribution for XL as for S, M, and L. This robustness against early interruption of the optimization process in Step~2 can be seen as another indication for the scalability of our approach.

\begin{figure}[t]
    \centering
    \includegraphics[width=0.7\linewidth]{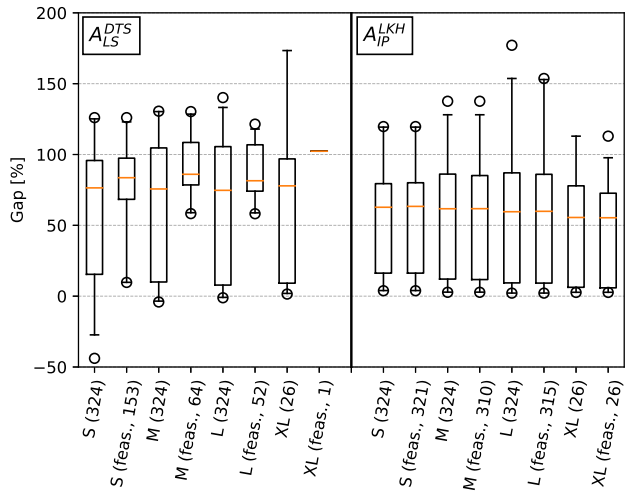}
    \caption{Gap to lower bound for newly generated instances. 
    For each algorithm, the results for all instances in the respective set and the results for only those instances in the set where the algorithm produced a feasible solution (indicated by `feas.') is shown. Numbers in parentheses indicate the respective number of instances.
    One outlier for the instance set XL was omitted for each of the two algorihtms.}
    \label{fig:lbgap}
\end{figure}

\subsubsection*{Running time}

The average running time fo instances in PBT was $0.42$ seconds for $A_{\text{LS}}^{\text{DTS}}$ and $0.34$ seconds for $A_{\text{IP}}^{\text{LKH}}$. The average running time for instances in S\&L was $7.71$ seconds for $A_{\text{LS}}^{\text{DTS}}$ and $15.12$ seconds for $A_{\text{IP}}^{\text{LKH}}$. The average running time for the instances in S, M, and L was $31.88$ seconds for $A_{\text{LS}}^{\text{DTS}}$ and $254.52$ seconds for $A_{\text{IP}}^{\text{LKH}}$, while, on average, the instances in XL required $294.09$ minutes to be solved with $A_{\text{LS}}^{\text{DTS}}$ and $132.61$ minutes to be solved with $A_{\text{IP}}^{\text{LKH}}$.

Each instance in PBT was solved in less than $1.30$ seconds by $A_{\text{LS}}^{\text{DTS}}$ and less than $1.50$ seconds by $A_{\text{IP}}^{\text{LKH}}$. Each instance in S\&L could be solved in less than $52.20$ seconds by $A_{\text{LS}}^{\text{DTS}}$ and less than $859.90$ seconds by $A_{\text{IP}}^{\text{LKH}}$, with $90\%$ of them actually being solved in less than $20.64$ seconds by $A_{\text{LS}}^{\text{DTS}}$ and less than $19.68$ seconds by $A_{\text{IP}}^{\text{LKH}}$.
Furthermore, $A_{\text{LS}}^{\text{DTS}}$ solved all of the instances in S, M, and L within the set time limits, while $A_{\text{IP}}^{\text{LKH}}$ solved $96.30\%$ of those instances within the time limit.
Among the instances in XL, $A_{\text{LS}}^{\text{DTS}}$ solved $33.33$\% of the instances within the time limit, while $A_{\text{IP}}^{\text{LKH}}$ solved $59.26\%$.

We point out that the instances in sets L and XL are considerably larger than those considered in the \LRVDC{} literature so far, with XL being an order of magnitude larger.
The main computational bottleneck of our method is Step~2, i.e., approximating the \CFL{} instance in case of $A_{\text{LS}}^{\text{DTS}}$ or solving IP~\eqref{ip_cfl} in case of $A_{\text{IP}}^{\text{LKH}}$. 
As pointed out earlier, we did not optimize our implementation for speed. Replacing the local search with another solution approach to \CFL{}, or speeding up the solution process for IP~\eqref{ip_cfl} by adding valid inequalities or employing a decomposition approach, respectively, could significantly speed up the solution process and put even larger instance sizes within reach.

Finally, we remark that the measured time for applying the double tree algorithm or LKH to compute the tours was negligible throughout our experiments.

\subsubsection*{Facility loads}

In our experiments, $918$ out of $1280$ ($71.72\%$) solutions computed by $A_{\text{LS}}^{\text{DTS}}$ exceed the capacity at at least one facility. In these solutions, on average, an overloaded facility serves an excessive demand of $13.24\%$ of its capacity on average and $73.50\%$ at maximum (on some of the instances where vehicle capacities are close to facility capacities). 
Over all solutions computed by $A_{\text{LS}}^{\text{DTS}}$, $27.27\%$ of the open facilities had to serve an excess of demand.

For $A_{\text{IP}}^{\text{LKH}}$ on the other hand, facility capacities were only exceeded in $27$ out $1280$ ($2.11\%$) instances, with an average excess of $3.79\%$ of capacity among all overloaded facilities and a worst-case of $15.25\%$. On average, only $0.22\%$ of the open facilities had to serve an excess of demand when applying this variant of the algorithm. 
We conclude that the use of the IP for assigning clusters to facilities results in feasible solutions in the vast majority of cases and that in the cases where it does result in feasible solutions, a small extension of the capacity of individual facilities suffices.

\subsubsection*{Comparing LS and IP} 
To measure the direct effect of the assignment method used in Step~2, we compare the variant using local search with that using the IP using the same routing heuristic, respectively.
Across all instances, the average increase in costs when using $A_{\text{IP}}^{\text{DTS}}$ with respect to the costs obtained when using $A_{\text{LS}}^{\text{DTS}}$ is $2.58\%$. Similarly, the average increase in costs when using $A_{\text{IP}}^{\text{LKH}}$ with respect to the costs obtained when using $A_{\text{LS}}^{\text{LKH}}$ is $2.76\%$. 
However, recall that most solutions produced by $A_{LS}^{X}$ are infeasible for $X \in \{DTS, LKH\}$ (feasibility is not affected by the routing heuristic). If we restrict our analysis to instances in which both $A_{LS}^{X}$ and $A_{IP}^{X}$ result in feasible solutions, the effect is reversed: across all instances, the average decrease in costs when using $A_{\text{IP}}^{\text{DTS}}$ with respect to the costs obtained when using $A_{\text{LS}}^{\text{DTS}}$ is $0.19\%$. The same average decrease in costs is obtained when using $A_{\text{IP}}^{\text{LKH}}$ instead of $A_{\text{LS}}^{\text{LKH}}$.

We conclude that, on average, costs are not significantly affected by the choice of assignment heuristic in Step~2. As we saw earlier, running times tend to increase when using $A_{IP}^{X}$ instead of $A_{LS}^{X}$. However, $A_{IP}^{X}$ computes feasible solutions in most cases, making it a viable candidate for use in practice. This difference can also be observed in \cref{fig:scatter}: Almost all solutions produced by $A_{\text{IP}}^{\text{DTS}}$ lie left to the $0$-excess line, or at least close to it.

\begin{figure}[t]
    \centering
    \includegraphics[width=0.7\linewidth]{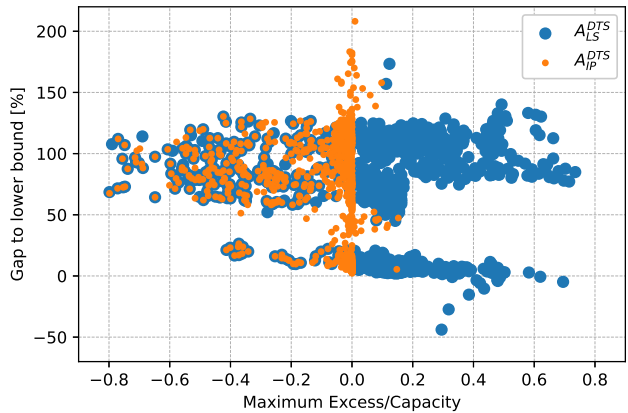}
    \caption{Maximum excessive load vs. gap to lower bound. Each point corresponds to a solution: All instances from all sets are separately solved with $A_{\text{LS}}^{\text{DTS}}$ and $A_{\text{IP}}^{\text{DTS}}$. The vertical coordinate denotes the gap to the lower bound for the solution. The horizontal coordinate denotes the largest excessive demand of a facility relative to its capacity in the respective solution. The same two outliers excluded from \cref{fig:lbgap} were excluded here as well.}
    \label{fig:scatter}
\end{figure}

\subsubsection*{Effect of post-optimization}

Across all instances and methods, the average routing cost improvement obtained after applying LKH is $8.20$\%.
As mentioned earlier, computation times for applying this post-optimization step are negligible.

\subsubsection*{Lower bounds}

In $70.9$\% of all instances, \CFL{} resulted in a better lower bound than MST. Moreover, \CFL{} lower bounds are on average $459.83\%$ larger than MST lower bounds. In general, one would expect \CFL{} to result in a better bound when facility costs are considerably higher than routing costs, in which case the costs of the \CFL{} solution should provide a good estimation of the actual \LRVDC{} costs. If routing costs are considerably higher than facility costs instead and when vehicle capacity utilization is low, MST is expected to provide a better estimates. Thus, a possible explanation for the better performance of the \CFL{} lower bound in our study is that facility costs are significantly higher than routing costs for most instances. In fact, this is also a plausible explanation for the clustering that can be observed in \cref{fig:scatter}: The lower cluster (below a threshold of approximately $25\%$ on the vertical axis) coincides with the instances for which facility costs are classified as ``large''. For larger facility costs, the lower bound provided by \CFL{} provides an increasingly better approximation to the optimal cost of the corresponding \LRVDC{} instance.

\subsubsection*{Summary of main findings}
The computational results for $A_{\text{LS}}^{\text{DTS}}$ indicate that the algorithm largely outperforms the worst-case guarantee proven in \cref{sec:algorithm}.
Moreover, using the heuristic improvements discussed in \cref{sec:variants}, variant $A_{\text{IP}}^{\text{LKH}}$ obtains solutions that do not exceed any facility capacities for the vast majority of instances. On the benchmark instances the produced solutions are within approximately 10\% of the best-known solutions on average. The experiments on newly generated instances provide evidence that these results can be extended to larger instance sizes as the performance relative to lower bounds remains the same from small instance sizes (comparable to the benchmark instances) to very large instance sizes. We conclude that the algorithmic approach can serve as a basis for a fast heuristic that allows computation of high-quality solutions even on very large instance sizes.

\section{Conclusions and Outlook}
\label{sec:conclusions}

In this paper, we proposed a bifactor approximation for \LRVDC{}, with worst-case guarantees for both capacity utilization and costs. The algorithm merges a minimum spanning tree in the input graph and a \CFL{} solution with respect to modified connection costs, which are both lower bounds for optimal \LRVDC{} solutions.

An interesting open question is the existence of a constant-factor approximation for \LRVDC{} that strictly adheres to the facility capacities.
The question for such an algorithm was already raised by \cite{RaviSinha2006}. \cite{chen2009approximation,chen2009cost} gave a positive answer for soft-capacitated variants of the problem. Our theoretical results can be seen as a natural next step towards an approximation algorithm strictly respecting hard capacities.
As demonstrated by the example given in the proof of \cref{lem:lower_bounds_not_enough}, however, such an algorithm would require the use of new and stronger lower bounds on the optimal solution value, which appear to be a natural starting point for future research. Furthermore, our computational experiments indicate that our approach can also serve as an efficient heuristic for use in practice. While the IP variant of the algorithm obtains feasible solutions in most cases, future research will hopefully reveal additional heuristic methods that allow to achieve feasibility for the remaining cases in which small capacity violations at some facilities still occur, thus providing implementable solutions for those application contexts in which small capacity extensions are not realizable.

\bibliography{bibliography_v1}

\begin{thebibliography}{33}
\providecommand{\natexlab}[1]{#1}
\providecommand{\url}[1]{\texttt{#1}}
\expandafter\ifx\csname urlstyle\endcsname\relax
  \providecommand{\doi}[1]{doi: #1}\else
  \providecommand{\doi}{doi: \begingroup \urlstyle{rm}\Url}\fi

\bibitem[Aggarwal et~al.(2013)Aggarwal, Louis, Bansal, Garg, Gupta, Gupta, and
  Jain]{AggarwalLouisBansalGargGuptaGuptaJain2013}
A.~Aggarwal, A.~Louis, M.~Bansal, N.~Garg, N.~Gupta, S.~Gupta, and S.~Jain.
\newblock A 3-approximation algorithm for the facility location problem with
  uniform capacities.
\newblock \emph{Mathematical Programming}, 141\penalty0 (1-2):\penalty0
  527--547, 2013.

\bibitem[Alpert et~al.(2003)Alpert, Kahng, Liu, Mandoiu, and
  Zelikovsky]{AlpertKahngMandoiuZelikovsky2003}
C.~J. Alpert, A.~B. Kahng, B.~Liu, I.~I. Mandoiu, and A.~Z. Zelikovsky.
\newblock Minimum buffered routing with bounded capacitive load for slew rate
  and reliability control.
\newblock \emph{IEEE Transactions on Computer-Aided Design of Integrated
  Circuits and Systems}, 22\penalty0 (3):\penalty0 241--253, 2003.

\bibitem[An et~al.(2017)An, Singh, and Svensson]{AnSinghSvensson2017}
H.-C. An, M.~Singh, and O.~Svensson.
\newblock {LP}-based algorithms for capacitated facility location.
\newblock \emph{SIAM Journal on Computing}, 46\penalty0 (1):\penalty0 272--306,
  2017.

\bibitem[Bansal et~al.(2012)Bansal, Garg, and Gupta]{BansalGargGupta2012}
M.~Bansal, N.~Garg, and N.~Gupta.
\newblock A 5-approximation for capacitated facility location.
\newblock In \emph{European Symposium on Algorithms}, pages 133--144. Springer,
  2012.

\bibitem[Barreto et~al.(2007)Barreto, Ferreira, Paixao, and
  Sousa~Santos]{BarretoFerreiraPaixaoSousasantos2007}
S.~Barreto, C.~Ferreira, J.~Paixao, and B.~Sousa~Santos.
\newblock Using clustering analysis in a capacitated location-routing problem.
\newblock \emph{European Journal of Operational Research}, 179\penalty0
  (3):\penalty0 968--977, 2007.

\bibitem[Byrka and Aardal(2010)]{ByrkaAardal2010}
J.~Byrka and K.~Aardal.
\newblock An optimal bifactor approximation algorithm for the metric
  uncapacitated facility location problem.
\newblock \emph{SIAM Journal on Computing}, 39\penalty0 (6):\penalty0
  2212--2231, 2010.

\bibitem[Chen and Chen(2009{\natexlab{a}})]{chen2009approximation}
X.~Chen and B.~Chen.
\newblock Approximation algorithms for soft-capacitated facility location in
  capacitated network design.
\newblock \emph{Algorithmica}, 53\penalty0 (3):\penalty0 263--297,
  2009{\natexlab{a}}.

\bibitem[Chen and Chen(2009{\natexlab{b}})]{chen2009cost}
X.~Chen and B.~Chen.
\newblock Cost-effective designs of fault-tolerant access networks in
  communication systems.
\newblock \emph{Networks: An International Journal}, 53\penalty0 (4):\penalty0
  382--391, 2009{\natexlab{b}}.

\bibitem[Christofides(1976)]{christofides1976worst}
N.~Christofides.
\newblock Worst-case analysis of a new heuristic for the travelling salesman
  problem.
\newblock Technical report, Carnegie-Mellon Univ Pittsburgh Pa Management
  Sciences Research Group, 1976.

\bibitem[Drexl and Schneider(2015)]{DrexlSchneider2015}
M.~Drexl and M.~Schneider.
\newblock A survey of variants and extensions of the location-routing problem.
\newblock \emph{European Journal of Operational Research}, 241\penalty0
  (2):\penalty0 283--308, 2015.

\bibitem[Golden et~al.(2008)Golden, Raghavan, and Wasil]{golden2008vehicle}
B.~L. Golden, S.~Raghavan, and E.~A. Wasil.
\newblock \emph{The vehicle routing problem: latest advances and new
  challenges}, volume~43.
\newblock Springer Science \& Business Media, 2008.

\bibitem[Harks et~al.(2013)Harks, K{\"o}nig, and
  Matuschke]{HarksKonigMatuschke2013}
T.~Harks, F.~G. K{\"o}nig, and J.~Matuschke.
\newblock Approximation algorithms for capacitated location routing.
\newblock \emph{Transportation Science}, 47\penalty0 (1):\penalty0 3--22, 2013.

\bibitem[Helsgaun(2000)]{helsgaun2000effective}
K.~Helsgaun.
\newblock An effective implementation of the {Lin--Kernighan} traveling
  salesman heuristic.
\newblock \emph{European Journal of Operational Research}, 126\penalty0
  (1):\penalty0 106--130, 2000.

\bibitem[Jain et~al.(2003)Jain, Mahdian, Markakis, Saberi, and
  Vazirani]{jain2003greedy}
K.~Jain, M.~Mahdian, E.~Markakis, A.~Saberi, and V.~V. Vazirani.
\newblock Greedy facility location algorithms analyzed using dual fitting with
  factor-revealing {LP}.
\newblock \emph{Journal of the ACM (JACM)}, 50\penalty0 (6):\penalty0 795--824,
  2003.

\bibitem[Karlin et~al.(2021)Karlin, Klein, and Gharan]{karlin2021slightly}
A.~R. Karlin, N.~Klein, and S.~O. Gharan.
\newblock A (slightly) improved approximation algorithm for metric tsp.
\newblock In \emph{Proceedings of the 53rd Annual ACM SIGACT Symposium on
  Theory of Computing}, pages 32--45, 2021.

\bibitem[Korupolu et~al.(2000)Korupolu, Plaxton, and
  Rajaraman]{KorupoluPlaxtonRajaraman2000}
M.~R. Korupolu, C.~G. Plaxton, and R.~Rajaraman.
\newblock Analysis of a local search heuristic for facility location problems.
\newblock \emph{Journal of Algorithms}, 37\penalty0 (1):\penalty0 146--188,
  2000.

\bibitem[Levi et~al.(2012)Levi, Shmoys, and Swamy]{LeviShmoysSwamy2012}
R.~Levi, D.~B. Shmoys, and C.~Swamy.
\newblock {LP}-based approximation algorithms for capacitated facility
  location.
\newblock \emph{Mathematical Programming}, 131\penalty0 (1-2):\penalty0
  365--379, 2012.

\bibitem[Li and Simchi-Levi(1990)]{LiSmichilevi1990}
C.-L. Li and D.~Simchi-Levi.
\newblock Worst-case analysis of heuristics for multidepot capacitated vehicle
  routing problems.
\newblock \emph{ORSA Journal on Computing}, 2\penalty0 (1):\penalty0 64--73,
  1990.

\bibitem[Li(2013)]{li20131}
S.~Li.
\newblock A 1.488 approximation algorithm for the uncapacitated facility
  location problem.
\newblock \emph{Information and Computation}, 222:\penalty0 45--58, 2013.

\bibitem[Maranzana(1964)]{Maranzana1964}
F.~Maranzana.
\newblock On the location of supply points to minimize transport costs.
\newblock \emph{Journal of the Operational Research Society}, 15\penalty0
  (3):\penalty0 261--270, 1964.

\bibitem[Morell and Skutella(2020)]{morell2020single}
S.~Morell and M.~Skutella.
\newblock Single source unsplittable flows with arc-wise lower and upper
  bounds.
\newblock In \emph{International Conference on Integer Programming and
  Combinatorial Optimization}, pages 294--306. Springer, 2020.

\bibitem[Prins et~al.(2006)Prins, Prodhon, and
  Wolfler~Calvo]{PrinsProdhonWolflercalvo2006}
C.~Prins, C.~Prodhon, and R.~Wolfler~Calvo.
\newblock Solving the capacitated location-routing problem by a {GRASP}
  complemented by a learning process and a path relinking.
\newblock \emph{4OR}, 4\penalty0 (3):\penalty0 221--238, 2006.

\bibitem[Prodhon and Prins(2014)]{ProdhonPrins2014}
C.~Prodhon and C.~Prins.
\newblock A survey of recent research on location-routing problems.
\newblock \emph{European Journal of Operational Research}, 238\penalty0
  (1):\penalty0 1--17, 2014.

\bibitem[Ravi and Sinha(2006)]{RaviSinha2006}
R.~Ravi and A.~Sinha.
\newblock Approximation algorithms for problems combining facility location and
  network design.
\newblock \emph{Operations Research}, 54\penalty0 (1):\penalty0 73--81, 2006.

\bibitem[Salhi and Nagy(1999)]{salhi1999consistency}
S.~Salhi and G.~Nagy.
\newblock Consistency and robustness in location-routing.
\newblock \emph{Studies in Locational Analysis}, \penalty0 (13):\penalty0
  3--19, 1999.

\bibitem[Salhi and Rand(1989)]{SalhiRand1989}
S.~Salhi and G.~K. Rand.
\newblock The effect of ignoring routes when locating depots.
\newblock \emph{European Journal of Operational Research}, 39\penalty0
  (2):\penalty0 150--156, 1989.

\bibitem[Schneider and Drexl(2017)]{SchneiderDrexl2017}
M.~Schneider and M.~Drexl.
\newblock A survey of the standard location-routing problem.
\newblock \emph{Annals of Operations Research}, 259\penalty0 (1-2):\penalty0
  389--414, 2017.

\bibitem[Schneider and L{\"o}ffler(2019)]{SchneiderLoffler2019}
M.~Schneider and M.~L{\"o}ffler.
\newblock Large composite neighborhoods for the capacitated location-routing
  problem.
\newblock \emph{Transportation Science}, 53\penalty0 (1):\penalty0 301--318,
  2019.

\bibitem[Shmoys and Tardos(1993)]{shmoys1993approximation}
D.~B. Shmoys and {\'E}.~Tardos.
\newblock An approximation algorithm for the generalized assignment problem.
\newblock \emph{Mathematical Programming}, 62\penalty0 (1):\penalty0 461--474,
  1993.

\bibitem[Toth and Vigo(2003)]{TothVigo2003}
P.~Toth and D.~Vigo.
\newblock The granular tabu search and its application to the vehicle-routing
  problem.
\newblock \emph{INFORMS Journal on Computing}, 15\penalty0 (4):\penalty0
  333--346, 2003.

\bibitem[Tuzun and Burke(1999)]{TuzunBurke1999}
D.~Tuzun and L.~I. Burke.
\newblock A two-phase tabu search approach to the location routing problem.
\newblock \emph{European Journal of Operational Research}, 116\penalty0
  (1):\penalty0 87--99, 1999.

\bibitem[Webb(1968)]{Webb1968}
M.~Webb.
\newblock Cost functions in the location of depots for multiple-delivery
  journeys.
\newblock \emph{Journal of the Operational Research Society}, 19\penalty0
  (3):\penalty0 311--320, 1968.

\bibitem[Williamson and Shmoys(2011)]{WilliamsonShmoys2011}
D.~P. Williamson and D.~B. Shmoys.
\newblock \emph{The Design of Approximation Algorithms}.
\newblock Cambridge university press, 2011.

\end{thebibliography}

\end{document}